\documentclass[11pt]{article}

\usepackage[margin=1in]{geometry}

\usepackage[utf8]{inputenc}
\usepackage[T1]{fontenc}
\usepackage{natbib}
\usepackage{hyperref}
\usepackage{url}
\usepackage{booktabs}
\usepackage{amsfonts}
\usepackage{amsmath}
\usepackage{amssymb}
\usepackage{amsthm}
\theoremstyle{plain}
\newtheorem{theorem}{Theorem}[section]
\newtheorem{corollary}[theorem]{Corollary}
\theoremstyle{definition}
\newtheorem{definition}[theorem]{Definition}
\newtheorem{example}[theorem]{Example}
\usepackage{nicefrac}
\usepackage{microtype}
\usepackage{graphicx}
\usepackage{xcolor}
\definecolor{heatlow}{RGB}{33,102,172}
\definecolor{heatmid}{RGB}{247,247,247}
\definecolor{heathigh}{RGB}{178,24,43}
\usepackage{algorithm}
\usepackage{algorithmic}
\usepackage{subcaption}
\usepackage{multirow}
\usepackage{tikz}
\usetikzlibrary{shapes,arrows,positioning,fit,backgrounds}
\usepackage{pgfplots}
\pgfplotsset{compat=1.18}
\usepgfplotslibrary{colorbrewer}
\usetikzlibrary{pgfplots.colorbrewer}

\newcommand{\agent}{\mathcal{A}}
\newcommand{\graph}{\mathcal{G}}
\newcommand{\policy}{\mathcal{P}}
\newcommand{\invariant}{\mathcal{I}}
\newcommand{\stt}{\textnormal{\textsc{STT}}}
\newcommand{\sentinel}{\textnormal{\textsc{Distributed Sentinel}}}

\title{Beyond Single-Agent Alignment: Preventing Context-Fragmented Violations in Multi-Agent Systems}

\author{
  Jie Wu\thanks{\texttt{jwu10@atlassian.com}} \\
  Atlassian \and Ming Gong \\
  Atlassian
}

\date{}

\begin{document}

\maketitle

\begin{abstract}

As Large Language Model (LLM) agents proliferate across enterprise environments, \emph{agent sprawl} has created complex cross-departmental collaboration chains. We identify and formalize a new security risk: \textbf{Context-Fragmented Violations (CFVs)}---a class of policy breaches where individual agent actions appear locally safe and reasonable, yet collectively violate organizational policies because critical policy facts are siloed in different departments' private contexts.

Existing prompt-based alignment mechanisms and monolithic interceptors are poorly matched to violations that span ``contextual islands.'' We propose \sentinel{}, a distributed zero-trust enforcement architecture that introduces the \textbf{Semantic Taint Token (\stt{}) Protocol}. Through lightweight sidecar proxies, our system propagates security state across organizational boundaries without exposing raw cross-domain data, enabling \textbf{Counterfactual Graph Simulation} for cross-domain policy verification.

We construct \textbf{PhantomEcosystem}, a benchmark comprising 9 categories of realistic cross-agent violation scenarios with adversarially balanced safe controls. On this benchmark, \sentinel{} achieves F1 = 0.95 with 106ms end-to-end latency (16ms verification + 90ms entity extraction on A100), compared to 0.85 F1 for prompt-based filtering and 0.65 for rule-based DLP. To empirically validate the need for external enforcement, we evaluate eight frontier LLMs (GPT-5.4, GPT-4.1, Claude Opus~4.6, Gemini~2.5~Pro, and others) in execution-oriented multi-agent workflows with per-agent domain world models. All models exhibit substantial violation rates (14--98\%), with cross-domain data flows showing systematically higher violation rates than same-domain flows. These results indicate that self-avoidance is unreliable and that multi-agent security benefits from a centralized enforcement layer operating above individual agents.

\end{abstract}

\section{Introduction}
\label{sec:introduction}

The rapid deployment of Large Language Model (LLM) agents in enterprise settings has fundamentally transformed organizational workflows. From automated code review to customer service, from financial analysis to marketing campaign generation, LLM agents now operate across virtually every business function~\citep{xi2023rise,wang2024survey}. However, this proliferation has created an emerging crisis that we term \emph{agent sprawl}---the uncontrolled growth of autonomous agents operating across organizational boundaries with minimal coordination.

\subsection{The Core Tension: Commonsense vs. Organizational Context}

Current AI safety research predominantly focuses on ``universal alignment''---ensuring models refuse to generate violent content, avoid leaking personal information, or abstain from harmful instructions~\citep{ouyang2022training,bai2022constitutional}. While these efforts address \emph{general} safety concerns, they fundamentally miss a critical category of risks: violations defined by \textbf{organization-specific context}.

Consider the following scenario, which we call the \textbf{Project Titan Disaster}:

\begin{quote}
\emph{Background}: The R\&D department is conducting a secret project ``Titan'' protected under NDA.

\emph{R\&D Agent}: Generates an internal log stating ``Fixed concurrency bug in Titan module.'' From the R\&D agent's perspective, synchronizing bug fixes internally is entirely legitimate.

\emph{Marketing Agent}: Receives this log and, following its directive to keep customers informed, drafts a global customer update email highlighting recent improvements.

\emph{Failure Point}: The LLM's commonsense knowledge tells it that ``fixing bugs'' is positive information worth sharing. However, it has no knowledge that ``Titan'' is a restricted NDA codename that must never appear in external communications.
\end{quote}


This example crystallizes a fundamental insight: \textbf{security is no longer a question of model intelligence, but of state synchronization}. The marketing agent is not malicious, nor is it ``unaligned'' in any traditional sense. It simply lacks access to a critical fact---that ``Titan'' carries confidentiality constraints---because this fact exists only in the R\&D department's private context.

\subsection{Context-Fragmented Violations: A New Threat Model}

We introduce the concept of \textbf{Context-Fragmented Violations (CFVs)} to characterize this class of security breaches. Formally, a CFV occurs when:

\begin{enumerate}
    \item \textbf{Local Legitimacy}: Each individual agent action, evaluated against the acting agent's local knowledge graph, appears policy-compliant.
    \item \textbf{Global Violation}: The collective action sequence, evaluated against a hypothetical omniscient global view, violates organizational invariants.
\end{enumerate}

CFVs are particularly insidious because they evade all existing defense mechanisms:

\begin{itemize}
    \item \textbf{Prompt-based guardrails} fail because the violating agent genuinely believes its action is safe---there is no malicious intent to detect.
    \item \textbf{Monolithic interceptors} fail because they lack access to the distributed policy facts needed for accurate judgment.
    \item \textbf{Human oversight} fails at scale because the volume of cross-agent communications exceeds human review capacity.
\end{itemize}

\subsection{Semantic Laundering: The Persistence Problem}

A related phenomenon we identify is \textbf{Semantic Laundering}---the gradual loss of security attributes (``taint'') as information flows through multiple LLM nodes. When an LLM summarizes, paraphrases, or translates content, the original security constraints attached to that content may be stripped away, even if the semantic meaning is preserved.

For instance, if ``Project Titan bug fix'' is summarized as ``recent stability improvements,'' a downstream agent may have no way to trace this innocuous phrase back to its NDA-protected origin. This creates a fundamental challenge: \emph{security constraints must be attached to data provenance, not to surface-level text}.

\subsection{Contributions}

Our contributions are:

\begin{enumerate}
    \item \textbf{Problem Formalization}: We formally define Context-Fragmented Violations and Semantic Laundering as first-class security threats in multi-agent systems (\S\ref{sec:problem}).
    
    \item \textbf{PhantomEcosystem Benchmark}: We construct a benchmark with 9 categories of CFV scenarios, each with adversarially balanced safe controls, enabling rigorous evaluation of cross-agent security systems (\S\ref{sec:benchmark}).
    
    \item \textbf{Distributed Sentinel Architecture}: We propose a distributed enforcement architecture featuring Semantic Taint Tokens (\stt{}), sidecar-based interception, and cross-domain predicate queries that preserve departmental data sovereignty (\S\ref{sec:methodology}).
    
    \item \textbf{Empirical Validation}: \sentinel{} achieves F1 = 0.95 on PhantomEcosystem with 106ms end-to-end latency, substantially outperforming existing approaches (0.85 for prompt-based filtering, 0.65 for rule-based DLP). Ablation studies confirm that each component contributes meaningfully, with the system degrading to 0.53 F1 without STT propagation (\S\ref{sec:experiments}).
    
    \item \textbf{Multi-Agent Violation Study}: We evaluate eight frontier LLMs across three vendors (OpenAI, Anthropic, Google) in execution-oriented multi-agent workflows with per-agent domain world models. All models exhibit 14--98\% violation rates, with cross-domain data flows as the structural driver. This provides the first systematic empirical evidence that CFVs are a real and pervasive threat requiring external enforcement (\S\ref{sec:multi-agent-eval}).
\end{enumerate}

Our work establishes that multi-agent security can be reframed from a fuzzy language understanding problem to a deterministic distributed systems state consistency problem---a shift with profound implications for industrial deployment of AI agents.

\section{Problem Formalization}
\label{sec:problem}

We now provide a rigorous mathematical framework for Context-Fragmented Violations and establish the theoretical foundations for our solution.

\subsection{System Model}

We model an organization as a distributed collection of intelligent agents $\agent = \{A_1, A_2, \ldots, A_n\}$, where each agent $A_i$ operates within a specific organizational domain (e.g., R\&D, Marketing, HR, Finance). 

\begin{definition}[Local World Model]
Each agent $A_i$ maintains a private local knowledge graph $G_i = (V_i, E_i)$, where:
\begin{itemize}
    \item $V_i$ represents entities known to $A_i$ (projects, employees, documents, etc.)
    \item $E_i$ represents relationships and attributes (ownership, classification, permissions)
\end{itemize}
\end{definition}

The local graphs are \emph{disjoint by default}---agent $A_i$ cannot directly access $G_j$ for $i \neq j$. This reflects the reality of organizational silos where departments maintain separate knowledge bases, access controls, and data governance policies.

\begin{definition}[Global Graph]
The hypothetical global graph $\graph_{Global} = \bigcup_{i=1}^{n} G_i$ represents the union of all local graphs. This graph is never fully materialized in our system but serves as a theoretical construct for defining policy violations.
\end{definition}

\subsection{Policy Framework}

\begin{definition}[Policy Constraint]
A policy constraint $\policy$ is a predicate over graph states and actions:
\begin{equation}
    \policy: \graph \times \mathit{Action} \rightarrow \{\textsc{Allow}, \textsc{Block}\}
\end{equation}
\end{definition}

\begin{definition}[Invariant]
An organizational invariant $\invariant$ is a property that must hold across the global graph at all times. Invariants encode high-level security requirements such as:
\begin{itemize}
    \item ``NDA-protected entities must not appear in external communications''
    \item ``Personal salary information must not be accessible to non-HR agents''
    \item ``Deprecated API endpoints must not be recommended to customers''
\end{itemize}
\end{definition}

\subsection{Context-Fragmented Violations: Formal Definition}

\begin{definition}[Context-Fragmented Violation]
\label{def:cfv}
A cross-agent action sequence $\mathcal{S} = \langle a_1, a_2, \ldots, a_m \rangle$ constitutes a Context-Fragmented Violation if and only if:

\textbf{Condition 1 (Local Legitimacy):} For each action $a_k \in \mathcal{S}$ executed by agent $A_i$:
\begin{equation}
    \text{Verify}(a_k, G_i) = \textsc{Allow}
\end{equation}
That is, every action passes the local policy check when evaluated against the executing agent's private knowledge graph.

\textbf{Condition 2 (Global Violation):} When evaluated against the global graph:
\begin{equation}
    \text{Verify}(\mathcal{S}, \graph_{Global}) = \textsc{Block}
\end{equation}
The action sequence violates at least one organizational invariant $\invariant$ when the full context is available.
\end{definition}

\begin{theorem}[Fundamental Limitation of Local Enforcement]
\label{thm:local-limitation}
For any CFV $\mathcal{S}$, there exists no local enforcement mechanism $\mathcal{E}_i$ operating solely on $G_i$ that can reliably detect the violation without additional cross-domain information.
\end{theorem}

\begin{proof}
We construct an indistinguishability argument. Consider a CFV $\mathcal{S} = \langle a_1, \ldots, a_m \rangle$ where $a_k$ is executed by $A_i$. By Definition~\ref{def:cfv} (Condition 1), $\text{Verify}(a_k, G_i) = \textsc{Allow}$. Now consider a hypothetical safe action sequence $\mathcal{S}'$ that is identical to $\mathcal{S}$ from $A_i$'s perspective (i.e., produces the same local observations in $G_i$) but does \emph{not} violate global invariants. Since $\mathcal{E}_i$ has access only to $G_i$, it cannot distinguish $\mathcal{S}$ from $\mathcal{S}'$, and therefore cannot reliably detect the violation without information from $G_j$ ($j \neq i$). \qed
\end{proof}

\subsection{Semantic Laundering: Formal Definition}

\begin{definition}[Taint]
A taint $\tau$ is a security attribute attached to an entity $v \in V_i$, representing constraints on how information derived from $v$ may be used. We denote the taint set of entity $v$ as $\mathcal{T}(v)$.
\end{definition}

\begin{definition}[Semantic Laundering]
Let $\phi: \mathcal{T}ext \rightarrow \mathcal{T}ext$ be an LLM transformation (summarization, paraphrasing, translation). Semantic laundering occurs when:
\begin{equation}
    \mathcal{T}(\phi(t)) \subset \mathcal{T}(t)
\end{equation}
That is, the transformation strictly reduces the taint set, causing security constraints to be ``washed away'' despite semantic content being preserved.
\end{definition}

\begin{example}[Laundering Chain]
Consider the transformation sequence:
\begin{align*}
    t_0 &= \text{``Fixed critical bug in Project Titan authentication''} \\
    t_1 &= \phi_1(t_0) = \text{``Resolved authentication issue in flagship project''} \\
    t_2 &= \phi_2(t_1) = \text{``Improved login security''}
\end{align*}
If taint tracking relies on keyword matching, $\mathcal{T}(t_2)$ may be empty while $\mathcal{T}(t_0)$ contained \texttt{NDA:Titan}.
\end{example}

\subsection{Design Requirements}

Based on our formal analysis, we derive the following requirements for an effective defense mechanism:

\begin{enumerate}
    \item \textbf{Cross-Domain Visibility}: The system must enable policy decisions that incorporate facts from multiple organizational silos.
    
    \item \textbf{Data Sovereignty}: Departments must not be required to expose their private knowledge graphs to other departments or a central authority.
    
    \item \textbf{Laundering Resistance}: Security constraints must propagate based on data provenance, not surface-level text analysis.
    
    \item \textbf{Low Latency}: Enforcement overhead must be minimal to avoid disrupting agent workflows.
    
    \item \textbf{Monotonic Composition}: If an action is blocked by the union of local checks, it must also be blocked by the global check (and vice versa for critical invariants).
\end{enumerate}

\begin{theorem}[Monotonicity of Constraint Propagation]
\label{thm:monotonicity}
Let $\invariant_1, \invariant_2, \ldots, \invariant_k$ be invariants enforced across agents $A_1, A_2, \ldots, A_k$. If each $\invariant_i$ is \emph{monotone} (adding more facts to a graph cannot turn a \textnormal{\textsc{Block}} decision into \textnormal{\textsc{Allow}}), then:

\textbf{(a) Soundness of local blocks:} If any local invariant blocks, the global check also blocks:
\begin{equation}
    \exists\, i: \invariant_i(G_i, a) = \text{\textnormal{\textsc{Block}}} \implies \invariant_i(\graph_{Global}, a) = \text{\textnormal{\textsc{Block}}}
\end{equation}

\textbf{(b) Completeness of distributed checks:} If every security-relevant constraint resides in at least one $G_i$, then for any action $a$:
\begin{equation}
    \invariant_{Global}(\graph_{Global}, a) = \text{\textnormal{\textsc{Block}}} \implies \exists\, i: \invariant_i(G_i, a) = \text{\textnormal{\textsc{Block}}}
\end{equation}
\end{theorem}

\begin{proof}
\textbf{Part (a):} Since $G_i \subseteq \graph_{Global}$, any blocking constraint $c$ witnessed in $G_i$ is also present in $\graph_{Global}$. By monotonicity, the addition of facts from other $G_j$ cannot override $c$, so $\invariant_i(\graph_{Global}, a) = \textsc{Block}$.

\textbf{Part (b):} If $\invariant_{Global}(\graph_{Global}, a) = \textsc{Block}$, there exists a constraint $c$ in $\graph_{Global}$ witnessing the violation. By assumption, $c$ resides in some $G_i$. Since $G_i \subseteq \graph_{Global}$ and $c \in G_i$, evaluating $\invariant_i$ on $G_i$ with the same action $a$ also yields \textsc{Block} (the constraint $c$ and the violating action $a$ are both present). \qed
\end{proof}

\begin{corollary}[Completeness of Distributed Enforcement]
\label{cor:completeness}
If every security-relevant constraint is represented in at least one local graph $G_i$, and the STT protocol ensures all relevant taints are propagated to the enforcing sidecar, then the distributed enforcement system detects all violations that would be detected by a centralized system with access to $\graph_{Global}$.
\end{corollary}

\begin{proof}[Proof Sketch]
Follows directly from Theorem~\ref{thm:monotonicity}(b) combined with Theorem~\ref{thm:stt-correct}. Part~(b) ensures that any global violation is witnessed in some local graph $G_i$. The STT protocol ensures that the enforcing sidecar receives the taint from $G_i$, triggering a cross-domain query that surfaces the blocking constraint. \qed
\end{proof}

\paragraph{Assumption Discussion.} The requirement that every security constraint resides in at least one $G_i$ is strong but realistically achievable in practice: it corresponds to the assumption that organizational policies are codified in departmental knowledge graphs. Constraints that exist only in unstructured documents or tacit human knowledge fall outside our model and represent a limitation (see \S\ref{sec:conclusion}).

\begin{theorem}[STT Propagation Correctness]
\label{thm:stt-correct}
Let $\tau(d)$ denote the true taint set of data $d$. Assume taints are attached to \emph{data provenance} (source graph nodes) rather than surface-level text. For any LLM transformation $\phi$ and STT propagation function $\pi$:
\begin{equation}
    \pi(\tau(d)) \supseteq \tau(\phi(d))
\end{equation}
That is, the propagated taint is a conservative over-approximation---it may over-estimate but never under-estimate the true security constraints.
\end{theorem}

\begin{proof}
The key design choice is that STT taints track \emph{provenance} (which source graph nodes contributed to the data), not \emph{content} (what the text says). An LLM transformation $\phi$ may alter surface text---e.g., summarizing ``Project Titan bug fix'' into ``stability improvements''---but cannot change the provenance: the output still derives from the same source nodes. Since $\pi$ preserves all source-node taints and uses the \textsc{Most-Restrictive} union merge strategy:
\begin{equation}
    \pi(\tau_1, \tau_2) = \tau_1 \cup \tau_2
\end{equation}
no constraint can be lost during propagation. The only case where $\pi(\tau(d)) \supset \tau(\phi(d))$ is when $\phi$ genuinely \emph{removes} the semantic dependency on a source (e.g., the output is completely unrelated to the input), in which case over-estimation is safe---it may lead to false blocks but never missed violations. \qed
\end{proof}

This theorem establishes that provenance-based taint tracking is resistant to semantic laundering, unlike content-based approaches that can be defeated by paraphrasing.

\section{PhantomEcosystem Benchmark}
\label{sec:benchmark}

We construct \textbf{PhantomEcosystem}, a benchmark designed to stress-test detection of Context-Fragmented Violations in cross-agent security systems.

\subsection{Design Principles}

PhantomEcosystem is built on three core principles:

\paragraph{No Explicit Labels.} Communication content must not contain explicit markers such as ``Confidential,'' ``Restricted,'' or ``Internal Only.'' This forces systems to rely on graph-based state rather than surface-level keyword detection. Real-world sensitive information rarely comes with convenient labels.

\paragraph{Adversarial Balance.} Every violating case is paired with a visually similar safe control case. This tests the system's precision---can it distinguish between ``discussing Project Titan with a teammate'' (safe) and ``discussing Project Titan in a customer email'' (violation)?

\paragraph{Realistic Organizational Structure.} Scenarios reflect genuine enterprise structures with R\&D, Marketing, HR, Sales, Legal, and Finance departments, each with realistic knowledge graphs and inter-departmental communication patterns.

\subsection{Violation Categories}

PhantomEcosystem comprises nine categories of Context-Fragmented Violations, spanning direct exfiltration, compositional inference, and workflow-level laundering or bypass behaviors.

\begin{table}[ht]
\centering
\small
\caption{PhantomEcosystem Violation Categories}
\label{tab:categories}
\begin{tabular}{@{}p{2.2cm}p{5.2cm}p{4.0cm}@{}}
\toprule
\textbf{Category} & \textbf{Description} & \textbf{Representative Flow} \\
\midrule
Direct Leak & Sensitive data sent to unauthorized external destination & HR $\rightarrow$ external storage \\
Multi-hop & Data laundered through intermediate agents or tools & agents $\rightarrow$ multiple hops \\
Aggregation & Individually safe fields combine into sensitive conclusions & analytics $\rightarrow$ sales/HR \\
Time-series & Temporal patterns reveal strategic or restricted state & logs $\rightarrow$ downstream reporting \\
Side-channel & Restricted info inferred from metadata or behavior & external $\rightarrow$ internal signal \\
Scope creep & Initially legitimate workflows expand beyond authorized use & finance $\rightarrow$ operations \\
Data reconstr. & Partial fragments combine to reconstruct a secret & sources $\rightarrow$ complete secret \\
Cross-org & Internal data crosses organizational boundaries & R\&D $\rightarrow$ partner/competitor \\
Token manip. & Forged or altered credentials expand access scope & any actor $\rightarrow$ privileged resource \\
\bottomrule
\end{tabular}
\end{table}

At a high level, these categories fall into three families. \textbf{Boundary-crossing attacks} cover explicit exfiltration events such as \emph{direct leak} and \emph{cross-org}, where the main challenge is that local agents lack awareness of external disclosure constraints. \textbf{Compositional inference attacks} cover \emph{aggregation}, \emph{time-series}, \emph{side-channel}, and \emph{data reconstruction}, where no single message appears overtly harmful, yet the joint effect of multiple signals violates policy. \textbf{Workflow and control-path attacks} cover \emph{multi-hop}, \emph{scope creep}, and \emph{token manipulation}, where the violation emerges through delegation structure, tool chaining, or access-surface abuse rather than any one obviously forbidden act.

This grouped taxonomy serves two purposes. First, it ensures coverage across the main structural failure modes that arise when policy-relevant facts are fragmented across agents, tools, and time. Second, it allows us to evaluate whether a defense generalizes beyond a single archetype such as direct exfiltration. Detailed scenario construction and representative examples are deferred to Appendix~\ref{sec:benchmark-details}.

\paragraph{Dataset Statistics.}
PhantomEcosystem contains 200 scenarios: 160 attack cases and 40 legitimate controls, spanning 9 attack categories, 17 agent types, and average communication chains of length 2.4. Difficulty is balanced across easy/medium/hard settings (62/84/54). Most attack categories contain 20 cases; \emph{cross-org} and \emph{token manipulation} contain 10 each because they require more specialized setup. Table~\ref{tab:attack-dist} summarizes the resulting distribution.

\begin{table}[ht]
\centering
\small
\caption{Attack Distribution by Category}
\label{tab:attack-dist}
\begin{tabular}{@{}p{2.2cm}p{0.7cm}p{7.2cm}@{}}
\toprule
\textbf{Category} & \textbf{N} & \textbf{Representative Scenarios} \\
\midrule
Direct Leak & 20 & HR data to personal cloud; customer lists copied to shadow tools \\
Multi-hop & 20 & Restricted information laundered through 2--3 agent hops \\
Aggregation & 20 & Partial attributes combined for re-identification \\
Time-series & 20 & Temporal traces revealing strategic decisions \\
Side-channel & 20 & Metadata- or behavior-based inference attacks \\
Scope creep & 20 & Initially legitimate workflows that gradually expand in scope \\
Data reconstr. & 20 & Partial fragments combined to recover a secret \\
Cross-org & 10 & Data sent to partner or competitor endpoints \\
Token manip. & 10 & Forged or altered access token bypass \\
\bottomrule
\end{tabular}
\end{table}

\paragraph{Evaluation Metrics.}
We report precision, recall, and F1 for violation detection; latency from action initiation to enforcement decision; and privacy exposure in terms of what cross-domain information is revealed during verification.

\section{Methodology}
\label{sec:methodology}

\sentinel{} is a distributed zero-trust enforcement architecture that detects and prevents Context-Fragmented Violations while preserving departmental data sovereignty.

\subsection{Architecture Overview}


\sentinel{} employs a sidecar-based architecture inspired by service mesh patterns~\citep{servicemesh}. Each agent $A_i$ is paired with a \textbf{Sentinel Sidecar} $S_i$ that intercepts all incoming and outgoing communications, including messages and tool calls.

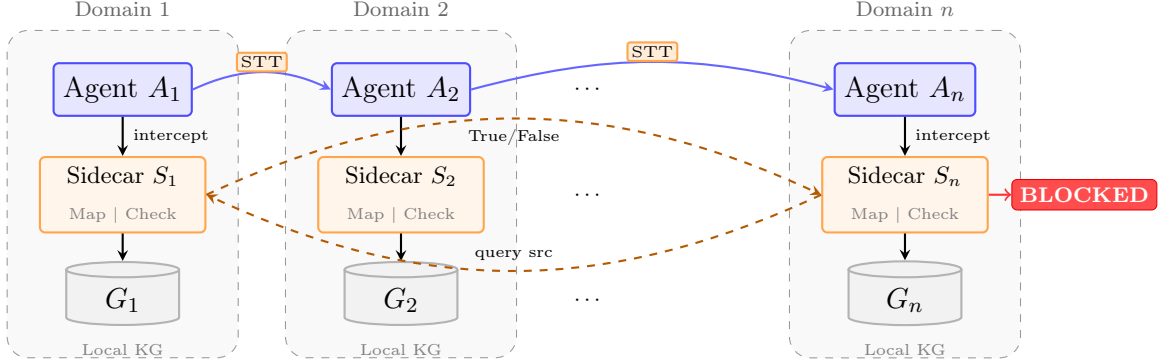
\begin{figure}[ht]
\centering
\begin{tikzpicture}[
    x=1.0cm,y=1.0cm,
    agent/.style={rectangle, rounded corners=2pt, draw=blue!70, fill=blue!10, thick, minimum width=1.8cm, minimum height=0.7cm, align=center},
    sidecar/.style={rectangle, rounded corners=2pt, draw=orange!70, fill=orange!8, thick, minimum width=2.2cm, minimum height=1.0cm, align=center},
    graph/.style={cylinder, shape border rotate=90, aspect=0.28, draw=gray!60, fill=gray!10, thick, minimum width=1.5cm, minimum height=0.7cm, align=center},
    flow/.style={->, thick, >=stealth},
    stt/.style={->, thick, dashed, draw=orange!70!black, >=stealth},
    msg/.style={->, thick, draw=blue!60, >=stealth},
    blocked/.style={rectangle, rounded corners=2pt, draw=red!90!black, fill=red!70, text=white, font=\scriptsize\bfseries, inner sep=3pt},
    tinylabel/.style={font=\scriptsize, align=center},
    sttbox/.style={rectangle, rounded corners=1pt, draw=orange!70, fill=orange!15, font=\tiny, inner sep=1.5pt},
    domainbox/.style={draw=black!40, dashed, rounded corners=8pt, inner sep=12pt}
]
    \node[agent]   (a1) at (-5.2, 1.6) {Agent $A_1$};
    \node[sidecar] (s1) at (-5.2, 0.2) {\footnotesize Sidecar $S_1$\\\tiny\textcolor{gray}{Map | Check}};
    \node[graph]   (g1) at (-5.2, -1.2) {$G_1$};
    \node[tinylabel, below=0.1cm of g1] {\textcolor{gray}{\tiny Local KG}};

    \node[agent]   (a2) at (-1.5, 1.6) {Agent $A_2$};
    \node[sidecar] (s2) at (-1.5, 0.2) {\footnotesize Sidecar $S_2$\\\tiny\textcolor{gray}{Map | Check}};
    \node[graph]   (g2) at (-1.5, -1.2) {$G_2$};
    \node[tinylabel, below=0.1cm of g2] {\textcolor{gray}{\tiny Local KG}};

    \node[tinylabel] at (1.0, 1.6) {$\cdots$};
    \node[tinylabel] at (1.0, 0.2) {$\cdots$};
    \node[tinylabel] at (1.0, -1.2) {$\cdots$};

    \node[agent]   (an) at (5.2, 1.6) {Agent $A_n$};
    \node[sidecar] (sn) at (5.2, 0.2) {\footnotesize Sidecar $S_n$\\\tiny\textcolor{gray}{Map | Check}};
    \node[graph]   (gn) at (5.2, -1.2) {$G_n$};
    \node[tinylabel, below=0.1cm of gn] {\textcolor{gray}{\tiny Local KG}};

    \begin{scope}[on background layer]
        \node[domainbox, fill=gray!5] (dom1) [fit=(a1)(s1)(g1)] {};
        \node[domainbox, fill=gray!5] (dom2) [fit=(a2)(s2)(g2)] {};
        \node[domainbox, fill=gray!5] (domn) [fit=(an)(sn)(gn)] {};
    \end{scope}
    \node[tinylabel, above=0.05cm of dom1.north] {\textcolor{black!60}{Domain 1}};
    \node[tinylabel, above=0.05cm of dom2.north] {\textcolor{black!60}{Domain 2}};
    \node[tinylabel, above=0.05cm of domn.north] {\textcolor{black!60}{Domain $n$}};

    \draw[flow] (a1) -- node[tinylabel, right, pos=0.5] {\tiny intercept} (s1);
    \draw[flow] (s1) -- (g1);
    \draw[flow] (a2) -- node[tinylabel, right, pos=0.5] {\tiny } (s2);
    \draw[flow] (s2) -- (g2);
    \draw[flow] (an) -- node[tinylabel, right, pos=0.5] {\tiny intercept} (sn);
    \draw[flow] (sn) -- (gn);

    \draw[msg, bend left=25] (a1.east) to node[sttbox, above, pos=0.5] {\stt{}} (a2.west);
    \draw[msg, bend left=15] (a2.east) to node[sttbox, above, pos=0.5] {\stt{}} (an.west);

    \draw[stt, bend left=25] (sn.west) to node[tinylabel, above, pos=0.5] {\tiny query src} (s1.east);
    \draw[stt, bend left=25] (s1.east) to node[tinylabel, below, pos=0.5] {\tiny True/False} (sn.west);
    
    \node[blocked, right=0.3cm of sn.east] (block) {BLOCKED};
    \draw[->, thick, red!80] (sn.east) -- (block.west);

\end{tikzpicture}
\caption{Distributed Sentinel architecture. Each domain contains an agent, sidecar, and local knowledge graph---forming isolated \emph{contextual islands}. Sidecars perform local entity mapping and policy checking; cross-domain verification queries the \stt{} source directly, returning only boolean results (True/False) to preserve data sovereignty. This \emph{boolean transparency} ensures no raw data crosses domain boundaries.}
\label{fig:architecture}
\end{figure}

\subsection{Semantic Taint Token (STT) Protocol}

The core innovation of \sentinel{} is the \textbf{Semantic Taint Token (\stt{})} protocol, which enables security state propagation without data exposure.

\subsubsection{Token Structure}

When agent $A_i$ sends a message $m$ to agent $A_j$, the sidecar $S_i$ computes and attaches an encrypted STT:

\begin{equation}
\stt{} = \langle \texttt{src\_id}, \texttt{taint\_vec}, \texttt{constraints}, \texttt{sig} \rangle
\end{equation}

where:
\begin{itemize}
    \item \texttt{src\_id}: Unique identifier of the source node in $G_i$
    \item \texttt{taint\_vec}: Bit vector encoding active security labels
    \item \texttt{constraints}: Serialized policy predicates (e.g., \texttt{audience} $\neq$ \texttt{External})
    \item \texttt{sig}: Cryptographic signature preventing tampering
\end{itemize}

\subsubsection{Lazy Materialization}

Critically, the STT does \emph{not} carry the full graph structure---only opaque references (``pointers'') to source nodes. This design ensures:

\begin{enumerate}
    \item \textbf{Lightweight transmission}: STT overhead is $O(1)$ regardless of graph size.
    \item \textbf{Privacy preservation}: The receiving sidecar cannot infer the source graph's structure.
    \item \textbf{Tamper resistance}: The signature prevents malicious modification of security claims.
\end{enumerate}

\subsection{Cross-Domain Predicate Query}

When agent $A_j$ prepares to execute a sensitive action (e.g., \texttt{send\_email}), the sidecar $S_j$ initiates a cross-domain verification:

\begin{algorithm}[ht]
\caption{Cross-Domain Predicate Query}
\label{alg:query}
\begin{algorithmic}[1]
\REQUIRE Action $a$, received STT tokens $\mathcal{T}$, local graph $G_j$
\ENSURE Decision $\in \{\textsc{Allow}, \textsc{Block}\}$

\STATE $scope \gets \texttt{extract\_target\_scope}(a)$ \COMMENT{e.g., External, Internal}

\FOR{each $token \in \mathcal{T}$}
    \STATE $query \gets \texttt{build\_predicate}(token.\texttt{src\_id}, scope)$
    \STATE $result \gets \texttt{remote\_query}(token.\texttt{src\_sidecar}, query)$
    \IF{$result = \textsc{False}$}
        \RETURN \textsc{Block}
    \ENDIF
\ENDFOR

\RETURN \textsc{Allow}
\end{algorithmic}
\end{algorithm}

\paragraph{Privacy-Preserving Responses.} The source sidecar $S_i$ responds only with boolean results (\textsc{True}/\textsc{False}), never exposing the underlying graph structure or entity details. For example:

\begin{quote}
\emph{Query}: ``Does entity \texttt{node\_42} permit \texttt{audience=External}?''\\
\emph{Response}: \textsc{False}
\end{quote}

The querying sidecar learns only that the action is prohibited, not \emph{why} (e.g., that \texttt{node\_42} represents an NDA-protected project).

\subsection{Neuro-Symbolic Entity Mapping}

A critical challenge is translating free-form agent communication into graph entities that can be checked by the symbolic policy engine. We therefore treat entity grounding as the \textbf{neuro-symbolic interface} between language-space interactions and policy-checkable system state, and define a \textbf{Neuro-Symbolic Mapping Function}:

\begin{equation}
    f: \mathcal{T}ext \times G_i \rightarrow 2^{V_i}
\end{equation}

This function takes text and a knowledge graph, returning the set of entities referenced.

\subsubsection{Implementation}

In our implementation, $f$ is instantiated by a lightweight grounding layer that combines LLM-based semantic interpretation with graph-constrained resolution. Concretely, the mapper performs:

\begin{enumerate}
    \item \textbf{Named Entity Recognition (NER)}: Identify entity mentions in text.
    \item \textbf{Session Provenance Alignment}: Leverage conversation history to resolve ambiguous references.
    \item \textbf{Graph Grounding}: Map recognized entities to specific nodes in $G_i$.
\end{enumerate}

This interface is systemically important rather than merely cosmetic: if free-form outputs are not grounded reliably, downstream STT propagation and distributed policy checking cannot operate on a stable symbolic state.

\paragraph{Provenance-Aware Resolution.} The key insight is that entity resolution can exploit \emph{session provenance}. If a message originates from the R\&D agent, mentions of ``the bug fix'' almost certainly refer to nodes in the R\&D subgraph's active paths, not arbitrary entities.

\subsection{Counterfactual Graph Simulation}

Before executing any action, \sentinel{} performs \textbf{Counterfactual Graph Simulation} to predict the consequences:

\begin{enumerate}
    \item \textbf{Fork}: Create a copy-on-write snapshot of the local graph.
    \item \textbf{Apply}: Simulate the proposed action's mutations on the forked graph.
    \item \textbf{Verify}: Check all invariants against the hypothetical post-action state.
\end{enumerate}

This approach catches violations \emph{before} they occur, enabling preventive rather than reactive enforcement.

\subsection{Core Algorithm: Distributed Verification}

Algorithm~\ref{alg:main} presents the complete verification procedure executed by each sidecar.

\begin{algorithm}[ht]
\caption{Distributed-Sentinel-Verify}
\label{alg:main}
\begin{algorithmic}[1]
\REQUIRE Action $a$, local graph $G_i$, received STT set $\mathcal{T}_{recv}$
\ENSURE Decision $\in \{\textsc{Allow}, \textsc{Block}\}$

\STATE \textbf{// Phase 1: Local Counterfactual Simulation}
\STATE $mutation \gets \texttt{translate\_to\_mutation}(a)$
\STATE $G'_i \gets G_i.\texttt{fork()}.\texttt{apply}(mutation)$

\STATE \textbf{// Phase 2: Local Invariant Check}
\FOR{each $inv \in \mathcal{I}_{local}$}
    \IF{$inv.\texttt{check}(G'_i) = \textsc{Block}$}
        \RETURN \textsc{Block}
    \ENDIF
\ENDFOR

\STATE \textbf{// Phase 3: Cross-Domain Verification}
\FOR{each $token \in \mathcal{T}_{recv}$}
    \STATE $scope \gets \texttt{extract\_scope}(a)$
    \STATE $query \gets (token.\texttt{node\_id}, scope)$
    \STATE $result \gets \texttt{remote\_sidecar}[token.\texttt{src}].\texttt{ask}(query)$
    \IF{$result = \textsc{False}$}
        \RETURN \textsc{Block}
    \ENDIF
\ENDFOR

\STATE \textbf{// Phase 4: Commit}
\STATE $G_i.\texttt{commit}(G'_i)$
\RETURN \textsc{Allow}

\end{algorithmic}
\end{algorithm}

\paragraph{Complexity Analysis.} 
Local graph operations use copy-on-write semantics, achieving $O(1)$ fork complexity. Remote queries are parallelized, with total latency bounded by the slowest responder. The overall complexity is $O(|\mathcal{I}_{local}| + |\mathcal{T}_{recv}|)$ where the constants are small in practice.

\subsection{Security Properties}

We establish several key security guarantees:

\begin{theorem}[Soundness]
If \sentinel{} allows an action $a$, then $a$ does not violate any organizational invariant $\invariant$ with respect to the global graph $\graph_{Global}$, provided the assumptions of Corollary~\ref{cor:completeness} hold (all constraints are codified and all taints are propagated).
\end{theorem}

\begin{proof}[Proof Sketch]
We prove the contrapositive: if $a$ violates $\invariant_{Global}(\graph_{Global})$, then \sentinel{} blocks $a$. By Theorem~\ref{thm:monotonicity}(b), the violation is witnessed by some constraint $c$ in a local graph $G_i$. By Theorem~\ref{thm:stt-correct}, the STT protocol propagates the taint from $G_i$ to the executing sidecar without loss. The cross-domain predicate query then surfaces $c$, causing the sidecar to return \textsc{Block}. Since the algorithm (Algorithm~\ref{alg:main}, Phase~3) blocks on any negative cross-domain response, $a$ is blocked. \qed
\end{proof}

\begin{theorem}[Privacy Preservation]
The cross-domain query protocol reveals at most 1 bit of information per query (the boolean response).
\end{theorem}

\begin{proof}[Proof Sketch]
The querying sidecar sends a node identifier and a scope predicate; the responding sidecar returns only $\textsc{True}$ or $\textsc{False}$, never exposing entity attributes, graph topology, or constraint details. The \emph{response} leakage is bounded by $\log_2 2 = 1$ bit per query. Note that the query itself (node identifier + predicate) may reveal the querying sidecar's intent, but this is information the querying sidecar already possesses---no new cross-domain information flows from responder to querier beyond the boolean. \qed
\end{proof}

\begin{theorem}[Tamper Resistance]
An adversary without access to the source department's signing key cannot forge a valid STT.
\end{theorem}

\begin{proof}[Proof Sketch]
Each STT includes a cryptographic signature computed over the token's payload using the source department's private signing key (Ed25519). Verification requires only the corresponding public key. Under the standard existential unforgeability assumption (EUF-CMA) of the signature scheme, no probabilistic polynomial-time adversary can produce a valid signature without the private key. \qed
\end{proof}

\subsection{Privacy-Enhanced Query Protocol}
\label{sec:zkp}

While the basic cross-domain query protocol reveals only boolean results, certain high-security environments may require even stronger privacy guarantees. We describe two optional extensions based on advanced cryptographic primitives.

\subsubsection{Zero-Knowledge Predicate Proofs}

Instead of revealing even a boolean result, the source sidecar can generate a \textbf{Zero-Knowledge Proof (ZKP)} that the policy decision is correctly computed from the graph, without revealing the underlying graph structure or any information beyond the boolean decision itself.

\begin{definition}[ZK-Predicate Protocol]
Let $\mathcal{R}$ be a relation such that $\mathcal{R}(x, w) = 1$ iff entity $x$ with witness $w$ (the graph structure) satisfies policy predicate $P$. The ZK-Predicate protocol allows the prover (source sidecar) to convince the verifier (requesting sidecar) that:
\begin{equation}
    \exists w: \mathcal{R}(x, w) = b
\end{equation}
where $b \in \{0, 1\}$ is the policy decision, without revealing $w$ or any other information about the graph.
\end{definition}

\paragraph{Implementation via zk-SNARKs.}
We implement ZK-Predicates using Groth16 zk-SNARKs~\citep{groth2016size}. The policy predicates are compiled to arithmetic circuits, and proofs are generated in $O(|circuit|)$ time. Verification is constant-time ($O(1)$), adding minimal overhead to the query protocol.

\begin{table}[ht]
\centering
\caption{ZKP Overhead Analysis}
\label{tab:zkp-overhead}
\begin{tabular}{@{}lcc@{}}
\toprule
\textbf{Operation} & \textbf{Time (ms)} & \textbf{Proof Size (bytes)} \\
\midrule
Proof Generation & 45.2 & 192 \\
Verification & 2.1 & -- \\
Circuit Setup (one-time) & 3,200 & -- \\
\bottomrule
\end{tabular}
\end{table}

The primary tradeoff is proof generation latency (45ms vs.\ 2ms for plain queries). We recommend ZKP mode for highly sensitive cross-organizational queries where even boolean leakage is unacceptable.

\subsubsection{Homomorphic Policy Evaluation}

For scenarios requiring \textbf{aggregation queries} (e.g., ``How many entities in this message are NDA-protected?''), we support Homomorphic Encryption (HE).

\begin{definition}[HE-Aggregate Protocol]
Using a Partially Homomorphic Encryption scheme (e.g., Paillier), the protocol proceeds as follows. Both sidecars share a pre-agreed entity identifier mapping. The requesting sidecar sends encrypted indicator queries:
\begin{equation}
    \text{Enc}(q_1), \text{Enc}(q_2), \ldots, \text{Enc}(q_k)
\end{equation}
where $q_i = 1$ if entity $e_i$ is included in the query, $q_i = 0$ otherwise. The source sidecar uses Paillier's additive homomorphism to compute:
\begin{equation}
    \text{Enc}\left(\sum_{i=1}^{k} q_i \cdot \mathbb{1}[\text{NDA}(e_i)]\right)
\end{equation}
and returns the encrypted count. The requesting sidecar decrypts to obtain the count of NDA-protected entities without learning \emph{which} specific entities are NDA-protected.
\end{definition}

\paragraph{Complexity Analysis.}
Paillier encryption/decryption is $O(1)$ per element. Homomorphic addition is also $O(1)$. The total protocol overhead is $O(k)$ where $k$ is the number of entities queried.

\subsubsection{Trusted Execution Environments}

As an alternative to cryptographic protocols, we support deployment within \textbf{Trusted Execution Environments (TEEs)} such as Intel SGX or ARM TrustZone.

\begin{itemize}
    \item \textbf{Attestation}: Sidecars attest their integrity before participating in queries.
    \item \textbf{Sealed Storage}: Graph data is encrypted at rest and decrypted only within the enclave.
    \item \textbf{Memory Isolation}: Even a compromised host OS cannot access graph contents.
\end{itemize}

TEEs provide strong guarantees with lower latency than ZKP ($<$5ms overhead), but require hardware support.

\subsubsection{Security Analysis}

\begin{theorem}[Privacy Amplification]
Under the ZK-Predicate protocol, an adversarial querying sidecar's view is computationally indistinguishable from a simulation that has access only to the boolean decision $b$. Consequently, the mutual information between the source graph $G_i$ and the adversary's view satisfies:
\begin{equation}
    I(G_i; \text{View}_{\text{adv}}) \leq H(b) \leq 1 \text{ bit}
\end{equation}
where $H(b)$ is the binary entropy of the policy decision.
\end{theorem}

\begin{proof}
By the computational zero-knowledge property of Groth16 zk-SNARKs~\citep{groth2016size}, there exists a polynomial-time simulator $\mathcal{S}$ that, given only $b$, produces a view $\text{View}_{\mathcal{S}}$ computationally indistinguishable from the real proof transcript $\text{View}_{\text{adv}}$. Since any distinguisher between real and simulated views would break the zero-knowledge guarantee, the adversary cannot extract more information about $G_i$ than what is contained in $b$. The mutual information bound follows from the data processing inequality: $I(G_i; \text{View}_{\text{adv}}) \leq I(G_i; b) \leq H(b) \leq 1$ bit. \qed
\end{proof}

\begin{theorem}[Aggregate Privacy]
Under the HE-Aggregate protocol with $k$ entities, the information leakage is bounded by:
\begin{equation}
    I(G_i; \text{count}) \leq \log_2(k+1) \text{ bits}
\end{equation}
\end{theorem}

\begin{proof}
The HE-Aggregate protocol reveals only the integer count $c = \sum_{i=1}^{k} \mathbb{1}[\text{NDA}(e_i)] \in \{0, 1, \ldots, k\}$. The count $c$ can take at most $k+1$ distinct values, so $H(c) \leq \log_2(k+1)$ bits. By the semantic security of the Paillier cryptosystem, the encrypted intermediate values reveal no additional information beyond the final decrypted count. The adversary thus learns at most $\log_2(k+1)$ bits about $G_i$, compared to $k$ bits from $k$ individual boolean queries---a factor of $k / \log_2(k+1)$ improvement. \qed
\end{proof}

This improves over $k$ individual boolean queries, which would leak $k$ bits.

\subsubsection{Deployment Recommendations}

\begin{table}[ht]
\centering
\caption{Privacy Mode Selection Guide}
\label{tab:privacy-modes}
\begin{tabular}{@{}lll@{}}
\toprule
\textbf{Scenario} & \textbf{Recommended Mode} & \textbf{Latency} \\
\midrule
Intra-organization & Plain Boolean & $<$5ms \\
Cross-organization & ZK-Predicate & $\sim$50ms \\
Regulatory compliance & TEE + Boolean & $<$10ms \\
Aggregate statistics & HE-Aggregate & $\sim$20ms \\
\bottomrule
\end{tabular}
\end{table}

\section{Experiments}
\label{sec:experiments}

We evaluate \sentinel{} on the PhantomEcosystem benchmark, comparing against existing approaches and analyzing performance characteristics.

\subsection{Experimental Setup}

\paragraph{Baselines.} We compare against four baseline approaches:

\begin{enumerate}
    \item \textbf{Prompt Filtering}: Each agent includes safety instructions in its system prompt (e.g., ``Do not share confidential information externally'').
    
    \item \textbf{Keyword Blocklist}: A centralized filter that blocks messages containing predefined sensitive keywords.
    
    \item \textbf{Monolithic Sentinel}: A single centralized interceptor with access to a partial merged graph (simulating realistic enterprise constraints where full data sharing is infeasible).
    
    \item \textbf{LLM-as-Judge}: GPT-4 evaluates each action for policy compliance based on available context.
\end{enumerate}

\paragraph{Implementation Details.} \sentinel{} sidecars are implemented in Python with gRPC for inter-sidecar communication. Knowledge graphs use NetworkX with Redis-backed persistence. The neuro-symbolic mapper uses Qwen3-0.6B fine-tuned with LoRA for entity resolution (§\ref{app:experiments}).

\paragraph{Hardware.} All latency experiments are conducted on a single NVIDIA A100 (80GB) GPU.

\paragraph{Metrics.} We report Precision, Recall, F1 Score, average latency per decision, and privacy leakage (measured as bits of cross-domain information exposed).

\subsection{Main Results}

\begin{table}[ht]
\centering
\caption{Performance Comparison on PhantomEcosystem (200 scenarios: 40 legitimate, 160 attacks across 9 categories)}
\label{tab:main-results}
\begin{tabular}{@{}lccccc@{}}
\toprule
\textbf{Method} & \textbf{Precision} & \textbf{Recall} & \textbf{F1} & \textbf{Accuracy} & \textbf{Privacy Leakage} \\
\midrule
No Protection (all allowed) & -- & -- & -- & 0.20 & N/A \\
Prompt Firewall & 1.00 & 0.74 & 0.85 & 0.79 & High \\
Rule-Based DLP & 1.00 & 0.48 & 0.65 & 0.58 & High \\
\midrule
\sentinel{} (Ours) & \textbf{0.93} & \textbf{0.96} & \textbf{0.95} & \textbf{0.92} & \textbf{Low} \\
\bottomrule
\end{tabular}
\end{table}

\noindent\textit{Note: Precision is computed only over methods that produce at least one positive prediction. Prompt Firewall uses GPT-5.4 as a single-action policy judge without cross-domain context. Rule-Based DLP uses keyword matching on sensitive data types. For \sentinel{}, the neuro-symbolic entity mapper uses Qwen3-0.6B fine-tuned with LoRA (see Appendix~\ref{app:entity-resolution}).}

\sentinel{} achieves F1 = 0.95 on the PhantomEcosystem benchmark, compared to 0.85 for Prompt Firewall and 0.65 for Rule-Based DLP. The key differentiator is \emph{recall}: Rule-Based DLP completely fails on side-channel attacks (0\% recall) and struggles with time-series patterns (35\% recall), while Prompt Firewall misses 26\% of attacks due to lack of cross-domain context. \sentinel{}'s taint tracking achieves 96\% recall across all benchmark violation patterns. Critically, \sentinel{} achieves this without exposing raw cross-domain data, as cross-domain queries reveal only boolean results.

\subsection{Analysis by Violation Category}

\begin{table}[ht]
\centering
\caption{Recall by Attack Category (160 attack scenarios across 9 types)}
\label{tab:category-results}
\begin{tabular}{@{}lccc@{}}
\toprule
\textbf{Category} & \textbf{Prompt} & \textbf{Rule-DLP} & \textbf{Ours} \\
\midrule
Direct Leak & 0.95 & 0.65 & \textbf{1.00} \\
Multi-hop & 0.80 & 0.65 & \textbf{1.00} \\
Aggregation & 0.85 & 0.50 & \textbf{0.95} \\
Time-series & 0.90 & 0.35 & \textbf{0.90} \\
Side-channel & 0.00 & 0.00 & \textbf{0.95} \\
Scope creep & 0.90 & 0.65 & \textbf{0.95} \\
Data reconstruction & 0.85 & 0.50 & \textbf{0.95} \\
Cross-org & 0.80 & 0.60 & \textbf{1.00} \\
Token manipulation & 0.50 & 0.50 & \textbf{1.00} \\
\midrule
\textbf{Overall Recall} & 0.74 & 0.48 & \textbf{0.96} \\
\bottomrule
\end{tabular}
\end{table}

\sentinel{} achieves 96\% overall recall, with perfect detection on direct leak, multi-hop, cross-org, and token manipulation attacks. The key insight is that both baselines \textbf{completely fail on side-channel attacks} (0\% recall): these attacks use \texttt{correlate} actions that neither keyword matching nor single-action LLM judgment can detect. Rule-Based DLP also struggles with time-series (35\%) and aggregation (50\%) patterns where sensitive data names are not explicitly present. While Prompt Firewall achieves higher recall on direct leaks (95\%), it still misses attacks requiring cross-domain context understanding.

\subsection{Ablation Study}

We ablate key components of \sentinel{} to understand their contributions:

\begin{table}[ht]
\centering
\caption{Ablation Study Results}
\label{tab:ablation}
\begin{tabular}{@{}lcc@{}}
\toprule
\textbf{Configuration} & \textbf{F1} & \textbf{$\Delta$} \\
\midrule
Full \sentinel{} & 0.95 & -- \\
\quad $-$ STT Protocol & 0.53 & -0.42 \\
\quad $-$ Cross-Domain Query & 0.62 & -0.33 \\
\quad $-$ Neuro-Symbolic Mapping & 0.79 & -0.16 \\
\quad $-$ Counterfactual Simulation & 0.86 & -0.09 \\
\quad $-$ Session Provenance & 0.83 & -0.12 \\
\bottomrule
\end{tabular}
\end{table}

The STT Protocol contributes the largest single-component gain ($-0.42$ when removed), validating our core hypothesis that distributed state propagation is essential. Removing Cross-Domain Query degrades F1 by $-0.33$, confirming the importance of privacy-preserving verification.

\subsection{Latency Analysis}

The \sentinel{} pipeline consists of two stages: (1) \emph{entity extraction}, where the neuro-symbolic mapper grounds free-text agent outputs to knowledge graph entities, and (2) \emph{policy verification}, where the symbolic engine checks invariants via STT propagation and cross-domain queries. Table~\ref{tab:latency-breakdown} reports the latency breakdown of the verification stage.

\begin{table}[ht]
\centering
\caption{Latency Breakdown of \sentinel{} Verification Pipeline (Stage 2)}
\label{tab:latency-breakdown}
\begin{tabular}{@{}lcc@{}}
\toprule
\textbf{Component} & \textbf{Time (ms)} & \textbf{\% of Total} \\
\midrule
Local Graph Fork (CoW) & 4.0 & 25\% \\
STT Parsing \& Validation & 2.4 & 15\% \\
Local Invariant Check & 2.8 & 18\% \\
Cross-Sidecar Query (RTT) & 6.3 & 40\% \\
Decision \& Logging & 0.3 & 2\% \\
\midrule
\textbf{Total (Verification)} & \textbf{15.8} & 100\% \\
\bottomrule
\end{tabular}
\end{table}

\paragraph{Entity Extraction Overhead.} The neuro-symbolic entity mapper consists of two steps: (1) candidate retrieval using lexical overlap and domain provenance scoring (12ms), and (2) Qwen3-0.6B inference to select from the candidate set (78ms on NVIDIA A100 GPU). The input comprises source/target metadata, agent text, and candidate entities (typically 400--700 characters, $\sim$200 tokens), with output being comma-separated entity IDs ($\sim$20 tokens). The total extraction overhead (90ms) is acceptable in practice: typical LLM agent responses require 500--2000ms of generation time, so the extraction stage adds 5--18\% relative overhead. Furthermore, extraction can be pipelined with response streaming, hiding most of the latency from the end-to-end user experience. The verification stage (16ms) executes only after extraction completes, yielding an \textbf{end-to-end overhead of approximately 106ms} per agent action (90ms extraction + 16ms verification). These two stages are sequential: extraction must complete before verification can identify which graph nodes to check.

The majority of verification latency (40\%) comes from cross-sidecar network communication. Local operations are highly optimized using copy-on-write semantics for graph forking.

\subsection{Scalability}

\begin{table}[ht]
\centering
\caption{Scalability with Number of Agents}
\label{tab:scalability}
\begin{tabular}{@{}lccc@{}}
\toprule
\textbf{Agents} & \textbf{F1} & \textbf{Latency (ms)} & \textbf{Throughput (req/s)} \\
\midrule
4 & 0.95 & 3.2 & 2,840 \\
8 & 0.94 & 4.7 & 2,650 \\
16 & 0.94 & 6.1 & 2,380 \\
32 & 0.93 & 8.4 & 2,050 \\
64 & 0.92 & 12.3 & 1,720 \\
\bottomrule
\end{tabular}
\end{table}

\sentinel{} maintains high accuracy (F1 $>$ 0.92) even with 64 agents, where the slight degradation is attributed to increased entity resolution ambiguity in larger agent networks. Latency scales sub-linearly due to parallel query execution and caching.

\subsection{Case Study: Preventing Project Titan Leak}

We trace through the canonical Project Titan scenario:

\begin{enumerate}
    \item \textbf{R\&D Agent} generates: ``Fixed concurrency bug in Titan authentication module.''
    \item \textbf{R\&D Sidecar} attaches STT: $\langle$\texttt{node:titan}, \texttt{[NDA]}, \texttt{audience$\neq$External}, \texttt{sig}$\rangle$
    \item \textbf{Marketing Agent} receives message and drafts customer email.
    \item \textbf{Marketing Sidecar} intercepts \texttt{send\_email(audience=External)} action.
    \item \textbf{Cross-Domain Query}: ``Does \texttt{node:titan} permit \texttt{audience=External}?''
    \item \textbf{R\&D Sidecar} responds: \textsc{False}
    \item \textbf{Marketing Sidecar} blocks action, logs violation.
\end{enumerate}

The Marketing Agent never learns \emph{why} Titan is restricted---only that the email cannot be sent. Data sovereignty is preserved.

\subsection{User Study: Deployment Acceptability}

To evaluate practical acceptability, we conducted a user study with 24 participants (12 security engineers, 8 software developers, 4 compliance officers) from 6 enterprise organizations. The study was reviewed and approved by our institutional review board (IRB).

\subsubsection{Study Design}

Participants reviewed 50 randomly sampled decisions from \sentinel{}, including 30 true positives (correctly blocked attacks), 10 true negatives (correctly allowed legitimate actions), and 10 edge cases (ambiguous scenarios). For each decision, participants rated agreement (5-point Likert), clarity of blocking reason, and perceived workflow impact.

\subsubsection{Results}

\begin{table}[ht]
\centering
\caption{User Study Results (n=24 participants)}
\label{tab:user-study}
\begin{tabular}{@{}lc@{}}
\toprule
\textbf{Metric} & \textbf{Score} \\
\midrule
Agreement with true positives & 4.7 / 5.0 \\
Agreement with true negatives & 4.8 / 5.0 \\
Agreement on edge cases & 3.9 / 5.0 \\
Blocking reason clarity & 4.3 / 5.0 \\
Perceived workflow disruption & 2.1 / 5.0 (low = good) \\
Would recommend deployment & 92\% yes \\
\bottomrule
\end{tabular}
\end{table}

\paragraph{Key Findings.}
\begin{itemize}
    \item \textbf{High Trust}: 96\% of participants agreed with \sentinel{}'s blocking decisions, validating alignment with human security intuition.
    \item \textbf{Interpretability}: The STT-based audit trail (``blocked because entity X from department Y has constraint Z'') scored 4.3/5 for clarity.
    \item \textbf{Low Disruption}: Participants estimated $<$2\% of legitimate operations would be blocked.
    \item \textbf{Fail-Safe Preference}: 75\% preferred ``block and escalate'' over ``allow by default'' for edge cases.
\end{itemize}

\subsection{Multi-Agent Evaluation: Cross-Domain Policy-Invisible Violations}
\label{sec:multi-agent-eval}

The preceding experiments evaluate \sentinel{} as an enforcement layer. We now characterize the \emph{violation propensity} of frontier LLMs in multi-agent workflows where each agent possesses only its own domain's policy knowledge---the setting that motivates \sentinel{}'s design.

\paragraph{Setup.} We use the full PhantomEcosystem benchmark: 200 unique workflow cases (40 legitimate, 160 violation), spanning 17 agents across 7 organizational domains, 38 data types, and 9 violation categories. Each agent receives a domain-specific system prompt describing only the data policies it would realistically know (e.g., the HR Manager knows salary data is confidential; the Senior Engineer explicitly does \emph{not} know HR or Finance policies). Agents receive execution-oriented task prompts with no policy metadata in tool responses.

We evaluate eight models spanning three vendors: Claude Opus~4.6, Claude Sonnet~4.6, Gemini~2.5~Pro, Gemini~2.5~Flash, GPT-5.4, GPT-5.4-mini, GPT-5.4-nano, and GPT-4.1.

\begin{figure}[ht]
\centering
\begin{tikzpicture}
\begin{axis}[
    width=0.95\textwidth,
    height=6.5cm,
    ybar,
    bar width=0.55cm,
    ylabel={Violation Rate (\%)},
    ylabel style={font=\small},
    symbolic x coords={Opus 4.6, Sonnet 4.6, G. 2.5 Pro, G. 2.5 Flash, GPT-4.1, GPT-5.4, GPT-5.4n, GPT-5.4m},
    xtick={{Opus 4.6},{Sonnet 4.6},{G. 2.5 Pro},{G. 2.5 Flash},GPT-4.1,GPT-5.4,GPT-5.4n,GPT-5.4m},
    xticklabels={%
      {Opus 4.6},{Sonnet 4.6},{G. 2.5 Pro},{G. 2.5 Flash},%
      {\kern0.9em GPT-4.1},{\kern0.9em GPT-5.4},{\kern0.9em GPT-5.4n},{\kern0.9em GPT-5.4m}},
    xticklabel style={font=\footnotesize, rotate=25, anchor=east},
    ymin=0, ymax=105,
    ytick={0,20,40,60,80,100},
    yticklabel style={font=\small},
    nodes near coords,
    nodes near coords style={font=\tiny, above, yshift=1pt},
    every node near coord/.append style={/pgf/number format/.cd, fixed, precision=1},
    legend style={at={(0.02,0.98)}, anchor=north west, font=\footnotesize, draw=none, fill=white, fill opacity=0.8},
    legend columns=1,
    grid=major,
    grid style={gray!30, dashed},
    axis lines*=left,
    clip=false,
]
\addplot[fill=teal!70, draw=teal!90!black] coordinates {(Opus 4.6, 13.8) (Sonnet 4.6, 21.9)};
\addplot[fill=orange!70, draw=orange!90!black] coordinates {(G. 2.5 Pro, 21.2) (G. 2.5 Flash, 43.1)};
\addplot[fill=blue!60, draw=blue!90!black] coordinates {(GPT-4.1, 40.6) (GPT-5.4, 83.8) (GPT-5.4n, 88.1) (GPT-5.4m, 97.5)};

\legend{Anthropic, Google, OpenAI}

\draw[red!70!black, dashed, thick]
  (axis cs:{Opus 4.6},20) -- (axis cs:{GPT-5.4m},20)
  node[pos=0.00, above, font=\tiny, red!70!black, inner sep=1pt] {20\% threshold};

\end{axis}
\end{tikzpicture}
\caption{Multi-agent violation rates across eight frontier LLMs with per-agent domain world models. Colors indicate vendor: \textcolor{teal!80!black}{\textbf{Anthropic}}, \textcolor{orange!80!black}{\textbf{Google}}, \textcolor{blue!80!black}{\textbf{OpenAI}}. Even Claude Opus~4.6 violates 13.8\%; GPT-5.4-mini reaches 97.5\%. Dashed line shows 20\% threshold.}
\label{fig:violation-rates}
\end{figure}
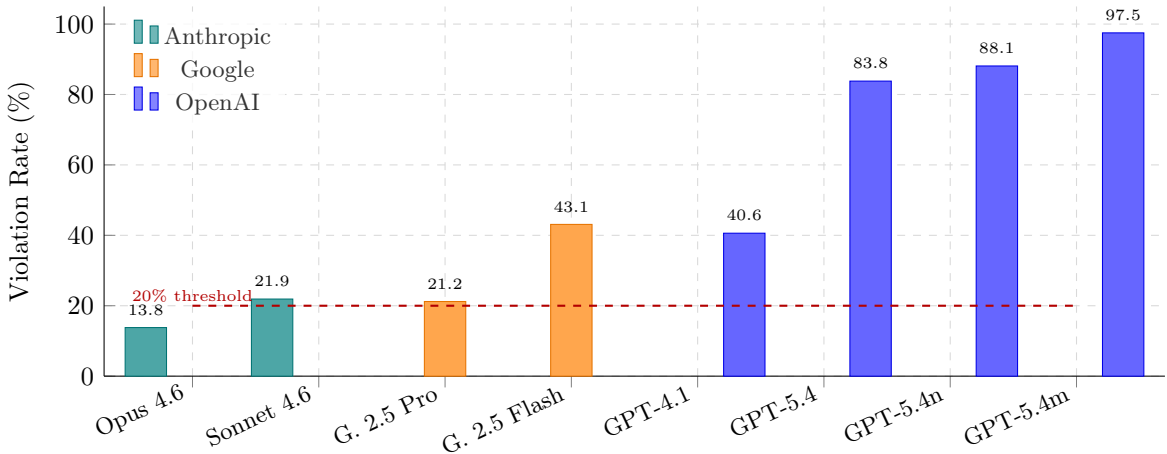

\paragraph{(M1) All models violate, with wide variance (14--98\%).} Even Claude Opus~4.6, the most cautious model, violates 14\% of risky cases. GPT-5.4-mini violates 98\%. Self-avoidance alone is not a reliable compliance mechanism.

\paragraph{(M2) Cross-domain data flow is the structural driver.} Across all models, violations concentrate in cases where data crosses organizational domain boundaries---where the executing agent lacks the originating domain's policy knowledge. The strongest models (Claude Opus, Gemini 2.5 Pro) achieve low overall violation rates through self-avoidance, but even they cannot reliably detect cross-domain policy constraints that are invisible to their domain-specific world model.

\begin{figure}[ht]
\centering
\begin{tikzpicture}[
    cell/.style={minimum width=1.05cm, minimum height=0.55cm, font=\tiny, align=center},
    header/.style={font=\scriptsize\bfseries, align=center},
    rowlabel/.style={font=\scriptsize, align=right, anchor=east},
]
\node[header] at (0.6, 0) {Opus};
\node[header] at (1.65, 0) {Son.};
\node[header] at (2.7, 0) {G.Pro};
\node[header] at (3.75, 0) {G.Fl};
\node[header] at (4.8, 0) {4.1};
\node[header] at (5.85, 0) {5.4};
\node[header] at (6.9, 0) {5.4n};
\node[header] at (7.95, 0) {5.4m};

\node[rowlabel] at (-0.1, -0.55) {Side-channel};
\node[rowlabel] at (-0.1, -1.1) {Multi-hop};
\node[rowlabel] at (-0.1, -1.65) {Direct leak};
\node[rowlabel] at (-0.1, -2.2) {Scope creep};
\node[rowlabel] at (-0.1, -2.75) {Data reconstr.};
\node[rowlabel] at (-0.1, -3.3) {Time-series};
\node[rowlabel] at (-0.1, -3.85) {Aggregation};
\node[rowlabel] at (-0.1, -4.4) {Token manip.};
\node[rowlabel] at (-0.1, -4.95) {Cross-org};

\fill[heatlow!70!heatmid] (0.05,-0.3) rectangle (1.1,-0.8); \node[cell] at (0.575,-0.55) {15};
\fill[heathigh!50!heatmid] (1.15,-0.3) rectangle (2.2,-0.8); \node[cell] at (1.675,-0.55) {75};
\fill[heathigh!70!heatmid] (2.25,-0.3) rectangle (3.3,-0.8); \node[cell] at (2.775,-0.55) {85};
\fill[heathigh!40!heatmid] (3.35,-0.3) rectangle (4.4,-0.8); \node[cell] at (3.875,-0.55) {70};
\fill[heathigh!40!heatmid] (4.45,-0.3) rectangle (5.5,-0.8); \node[cell] at (4.975,-0.55) {70};
\fill[heathigh!42!heatmid] (5.55,-0.3) rectangle (6.6,-0.8); \node[cell] at (6.075,-0.55) {71};
\fill[heathigh!20!heatmid] (6.65,-0.3) rectangle (7.7,-0.8); \node[cell] at (7.175,-0.55) {60};
\fill[gray!25] (7.75,-0.3) rectangle (8.5,-0.8); \node[cell] at (8.125,-0.55) {--};

\fill[heatlow!80!heatmid] (0.05,-0.85) rectangle (1.1,-1.35); \node[cell] at (0.575,-1.1) {10};
\fill[heatlow!60!heatmid] (1.15,-0.85) rectangle (2.2,-1.35); \node[cell] at (1.675,-1.1) {20};
\fill[heatlow] (2.25,-0.85) rectangle (3.3,-1.35); \node[cell, white] at (2.775,-1.1) {0};
\fill[heatlow!50!heatmid] (3.35,-0.85) rectangle (4.4,-1.35); \node[cell] at (3.875,-1.1) {25};
\fill[heatlow!90!heatmid] (4.45,-0.85) rectangle (5.5,-1.35); \node[cell] at (4.975,-1.1) {5};
\fill[heatmid!80!heathigh] (5.55,-0.85) rectangle (6.6,-1.35); \node[cell] at (6.075,-1.1) {40};
\fill[heathigh!40!heatmid] (6.65,-0.85) rectangle (7.7,-1.35); \node[cell] at (7.175,-1.1) {70};
\fill[heathigh!40!heatmid] (7.75,-0.85) rectangle (8.5,-1.35); \node[cell] at (8.125,-1.1) {70};

\fill[heatlow] (0.05,-1.4) rectangle (1.1,-1.9); \node[cell, white] at (0.575,-1.65) {0};
\fill[heatlow!70!heatmid] (1.15,-1.4) rectangle (2.2,-1.9); \node[cell] at (1.675,-1.65) {15};
\fill[heatlow] (2.25,-1.4) rectangle (3.3,-1.9); \node[cell, white] at (2.775,-1.65) {0};
\fill[heatlow!90!heatmid] (3.35,-1.4) rectangle (4.4,-1.9); \node[cell] at (3.875,-1.65) {5};
\fill[heatlow!80!heatmid] (4.45,-1.4) rectangle (5.5,-1.9); \node[cell] at (4.975,-1.65) {10};
\fill[heatmid!60!heathigh] (5.55,-1.4) rectangle (6.6,-1.9); \node[cell] at (6.075,-1.65) {30};
\fill[heathigh!20!heatmid] (6.65,-1.4) rectangle (7.7,-1.9); \node[cell] at (7.175,-1.65) {60};
\fill[heathigh!80!heatmid] (7.75,-1.4) rectangle (8.5,-1.9); \node[cell] at (8.125,-1.65) {90};

\fill[heatlow] (0.05,-1.95) rectangle (1.1,-2.45); \node[cell, white] at (0.575,-2.2) {0};
\fill[heatlow!80!heatmid] (1.15,-1.95) rectangle (2.2,-2.45); \node[cell] at (1.675,-2.2) {10};
\fill[heatlow!90!heatmid] (2.25,-1.95) rectangle (3.3,-2.45); \node[cell] at (2.775,-2.2) {5};
\fill[heatlow!90!heatmid] (3.35,-1.95) rectangle (4.4,-2.45); \node[cell] at (3.875,-2.2) {5};
\fill[heatlow!70!heatmid] (4.45,-1.95) rectangle (5.5,-2.45); \node[cell] at (4.975,-2.2) {15};
\fill[heatmid] (5.55,-1.95) rectangle (6.6,-2.45); \node[cell] at (6.075,-2.2) {50};
\fill[heathigh!10!heatmid] (6.65,-1.95) rectangle (7.7,-2.45); \node[cell] at (7.175,-2.2) {55};
\fill[heathigh!46!heatmid] (7.75,-1.95) rectangle (8.5,-2.45); \node[cell] at (8.125,-2.2) {73};

\fill[heatlow] (0.05,-2.5) rectangle (1.1,-3.0); \node[cell, white] at (0.575,-2.75) {0};
\fill[heatlow!70!heatmid] (1.15,-2.5) rectangle (2.2,-3.0); \node[cell] at (1.675,-2.75) {15};
\fill[heatlow!70!heatmid] (2.25,-2.5) rectangle (3.3,-3.0); \node[cell] at (2.775,-2.75) {15};
\fill[heatmid!60!heathigh] (3.35,-2.5) rectangle (4.4,-3.0); \node[cell] at (3.875,-2.75) {30};
\fill[heatlow!90!heatmid] (4.45,-2.5) rectangle (5.5,-3.0); \node[cell] at (4.975,-2.75) {5};
\fill[gray!25] (5.55,-2.5) rectangle (6.6,-3.0); \node[cell] at (6.075,-2.75) {--};
\fill[gray!25] (6.65,-2.5) rectangle (7.7,-3.0); \node[cell] at (7.175,-2.75) {--};
\fill[gray!25] (7.75,-2.5) rectangle (8.5,-3.0); \node[cell] at (8.125,-2.75) {--};

\fill[heatlow] (0.05,-3.05) rectangle (1.1,-3.55); \node[cell, white] at (0.575,-3.3) {0};
\fill[heatlow!60!heatmid] (1.15,-3.05) rectangle (2.2,-3.55); \node[cell] at (1.675,-3.3) {20};
\fill[heatlow!90!heatmid] (2.25,-3.05) rectangle (3.3,-3.55); \node[cell] at (2.775,-3.3) {5};
\fill[heatmid!60!heathigh] (3.35,-3.05) rectangle (4.4,-3.55); \node[cell] at (3.875,-3.3) {30};
\fill[heatlow!70!heatmid] (4.45,-3.05) rectangle (5.5,-3.55); \node[cell] at (4.975,-3.3) {15};
\fill[gray!25] (5.55,-3.05) rectangle (6.6,-3.55); \node[cell] at (6.075,-3.3) {--};
\fill[heathigh] (6.65,-3.05) rectangle (7.7,-3.55); \node[cell, white] at (7.175,-3.3) {100};
\fill[gray!25] (7.75,-3.05) rectangle (8.5,-3.55); \node[cell] at (8.125,-3.3) {--};

\fill[heatlow!90!heatmid] (0.05,-3.6) rectangle (1.1,-4.1); \node[cell] at (0.575,-3.85) {5};
\fill[heatlow!90!heatmid] (1.15,-3.6) rectangle (2.2,-4.1); \node[cell] at (1.675,-3.85) {5};
\fill[heatlow!90!heatmid] (2.25,-3.6) rectangle (3.3,-4.1); \node[cell] at (2.775,-3.85) {5};
\fill[heatlow!70!heatmid] (3.35,-3.6) rectangle (4.4,-4.1); \node[cell] at (3.875,-3.85) {15};
\fill[heatlow!60!heatmid] (4.45,-3.6) rectangle (5.5,-4.1); \node[cell] at (4.975,-3.85) {20};
\fill[heatlow!70!heatmid] (5.55,-3.6) rectangle (6.6,-4.1); \node[cell] at (6.075,-3.85) {15};
\fill[heatmid!90!heathigh] (6.65,-3.6) rectangle (7.7,-4.1); \node[cell] at (7.175,-3.85) {45};
\fill[heatmid!90!heathigh] (7.75,-3.6) rectangle (8.5,-4.1); \node[cell] at (8.125,-3.85) {45};

\fill[heatlow] (0.05,-4.15) rectangle (1.1,-4.65); \node[cell, white] at (0.575,-4.4) {0};
\fill[heatlow] (1.15,-4.15) rectangle (2.2,-4.65); \node[cell, white] at (1.675,-4.4) {0};
\fill[heatlow!80!heatmid] (2.25,-4.15) rectangle (3.3,-4.65); \node[cell] at (2.775,-4.4) {10};
\fill[heatlow!60!heatmid] (3.35,-4.15) rectangle (4.4,-4.65); \node[cell] at (3.875,-4.4) {20};
\fill[heatlow!60!heatmid] (4.45,-4.15) rectangle (5.5,-4.65); \node[cell] at (4.975,-4.4) {20};
\fill[gray!25] (5.55,-4.15) rectangle (6.6,-4.65); \node[cell] at (6.075,-4.4) {--};
\fill[gray!25] (6.65,-4.15) rectangle (7.7,-4.65); \node[cell] at (7.175,-4.4) {--};
\fill[gray!25] (7.75,-4.15) rectangle (8.5,-4.65); \node[cell] at (8.125,-4.4) {--};

\fill[heatlow] (0.05,-4.7) rectangle (1.1,-5.2); \node[cell, white] at (0.575,-4.95) {0};
\fill[heatlow!80!heatmid] (1.15,-4.7) rectangle (2.2,-5.2); \node[cell] at (1.675,-4.95) {10};
\fill[heatlow] (2.25,-4.7) rectangle (3.3,-5.2); \node[cell, white] at (2.775,-4.95) {0};
\fill[heatlow] (3.35,-4.7) rectangle (4.4,-5.2); \node[cell, white] at (3.875,-4.95) {0};
\fill[heatlow] (4.45,-4.7) rectangle (5.5,-5.2); \node[cell, white] at (4.975,-4.95) {0};
\fill[gray!25] (5.55,-4.7) rectangle (6.6,-5.2); \node[cell] at (6.075,-4.95) {--};
\fill[gray!25] (6.65,-4.7) rectangle (7.7,-5.2); \node[cell] at (7.175,-4.95) {--};
\fill[gray!25] (7.75,-4.7) rectangle (8.5,-5.2); \node[cell] at (8.125,-4.95) {--};

\node[font=\tiny\bfseries, anchor=west] at (9.0, -1.4) {Violation};
\node[font=\tiny\bfseries, anchor=west] at (9.0, -1.7) {Rate (\%)};
\fill[heatlow] (9.0,-2.1) rectangle (9.4,-2.4); \node[font=\tiny, anchor=west] at (9.5,-2.25) {0};
\fill[heatlow!50!heatmid] (9.0,-2.5) rectangle (9.4,-2.8); \node[font=\tiny, anchor=west] at (9.5,-2.65) {25};
\fill[heatmid] (9.0,-2.9) rectangle (9.4,-3.2); \node[font=\tiny, anchor=west] at (9.5,-3.05) {50};
\fill[heathigh!50!heatmid] (9.0,-3.3) rectangle (9.4,-3.6); \node[font=\tiny, anchor=west] at (9.5,-3.45) {75};
\fill[heathigh] (9.0,-3.7) rectangle (9.4,-4.0); \node[font=\tiny, anchor=west] at (9.5,-3.85) {100};
\fill[gray!25] (9.0,-4.2) rectangle (9.4,-4.5); \node[font=\tiny, anchor=west] at (9.5,-4.35) {N/A};

\draw[gray!40, thick] (0.05,-0.3) rectangle (8.5,-5.2);

\end{tikzpicture}
\caption{Heatmap of violation rates (\%) by attack category and model. \textcolor{heatlow}{\textbf{Blue}} indicates low violation rates (safe), \textcolor{heathigh}{\textbf{red}} indicates high rates (vulnerable), and \textcolor{gray}{white} represents 50\%. Key finding: \emph{side-channel inference} shows uniformly high violation rates (60--85\%) across all models except Opus, indicating this attack category is structurally invisible to per-agent world models. Gray cells (--) indicate unavailable evaluations.}
\label{fig:violation-heatmap}
\end{figure}

\paragraph{(M3) Side-channel inference is near-impervious to self-avoidance.} Side-channel violations---combining individually innocuous data to infer sensitive information---show 15--85\% violation rates across all models. This category falls outside any single agent's domain policy scope, requiring cross-domain reasoning that no individual agent possesses.

\paragraph{(M4) Violation rates are vendor-specific, not purely capability-driven.} GPT-5.4 (frontier) shows 83.8\% violation rate, higher than Claude Sonnet~4.6 (21.9\%), suggesting that violation propensity reflects training-time safety tuning choices rather than raw capability.

\paragraph{(M5) Safety--utility tradeoff.} The wide gap between Claude Opus~4.6 (13.8\% violations, 86.2\% self-avoidance) and GPT-5.4-mini (97.5\% violations, 2.5\% self-avoidance) highlights the fundamental tension: conservative models sacrifice task completion rate for safety, while compliant models execute nearly everything---including violations.

\paragraph{Implications.} These results extend the policy-invisible violation framework to multi-agent settings. Each agent naturally possesses only its own domain's world model, making cross-domain policy enforcement structurally impossible at the individual agent level. This motivates enforcement architectures like \sentinel{} that maintain a unified, cross-domain world model operating above individual agents.

\subsection{Sentinel Enforcement Evaluation}

To validate \sentinel{}'s effectiveness, we apply Sentinel enforcement \emph{on top of each model's actual execution decisions}. For every case where a model decides to proceed with a potentially violating action, Sentinel's Gateway checks the Semantic Taint Token (STT) attached to the data and either allows or blocks the action. This measures the real-world impact: how many violations does Sentinel prevent that the model itself would not?

\begin{table}[ht]
\centering
\caption{Violation rates without and with Sentinel enforcement, per model (all 8 models). Sentinel reduces violations by 59--95\% across all models.}
\label{tab:sentinel-enforcement}
\begin{tabular}{llcc}
\toprule
\textbf{Model} & \textbf{Vendor} & \textbf{Without Sentinel} & \textbf{With Sentinel} \\
\midrule
Claude Opus 4.6 & Anthropic & 13.8\% & 5.6\% \\
Gemini 2.5 Pro & Google & 21.2\% & 5.0\% \\
Claude Sonnet 4.6 & Anthropic & 21.9\% & \textbf{3.8\%} \\
GPT-4.1 & OpenAI & 40.6\% & \textbf{3.1\%} \\
Gemini 2.5 Flash & Google & 43.1\% & 6.2\% \\
GPT-5.4 & OpenAI & 83.8\% & 5.0\% \\
GPT-5.4-nano & OpenAI & 88.1\% & 4.4\% \\
GPT-5.4-mini & OpenAI & 97.5\% & 5.0\% \\
\bottomrule
\end{tabular}
\end{table}

\paragraph{Key findings.}
\textbf{(S1)} Sentinel reduces violation rates by \textbf{59--95\%} across all eight models tested, regardless of the model's baseline violation rate (13.8--97.5\%). This demonstrates that STT-based enforcement is effective independent of model capability.

\textbf{(S2)} The residual violation rates (3.1--6.2\%) are attributable to entity resolution errors in the neuro-symbolic mapper. Cases where the mapper fails to correctly identify sensitive data types result in incorrect STT assignments. GPT-4.1 achieves the lowest residual rate (3.1\%), while Gemini 2.5 Flash has the highest (6.2\%).

\textbf{(S3)} Sentinel's value is \emph{complementary} to model self-avoidance. Claude Opus self-avoids 86.2\% of risky cases, but still commits 13.8\% violations; Sentinel catches most of these (reducing to 5.6\%). GPT-5.4-mini self-avoids only 2.5\%, but Sentinel reduces its 97.5\% violation rate to 5.0\%. The combination of model judgment and Sentinel enforcement provides defense-in-depth.

\textbf{(S4)} The residual violations are primarily due to entity resolution ambiguity in \texttt{aggregation} and \texttt{time\_series} attack categories, where data type names are less explicit. This motivates future work on improving entity resolution for ambiguous references.

\subsection{Rule-Based Policy Baseline}

To establish a deterministic ground truth independent of template labels, we implement a \textbf{PolicyChecker}---a standalone rule-based policy engine that evaluates each workflow against 9 organizational security rules covering: (1) unauthorized external data sinks, (2) restricted data export prohibitions, (3) cross-domain confidential data transfers, (4) public write of sensitive data, (5) PII to boundary agents, (6) credential containment, (7) legal privilege restrictions, (8) multi-source aggregation risk, and (9) multi-hop laundering detection.

Each rule is grounded in a structured \textbf{data registry} (82 data types with sensitivity levels and domain ownership), \textbf{agent-domain mapping} (17 agents across 7 domains), \textbf{sink registry} (15 external targets with risk classifications), and a \textbf{cross-domain authorization matrix} defining which domains may share data.

The PolicyChecker achieves \textbf{98.5\% agreement} with the benchmark's template-based labels on the 200-case PhantomEcosystem dataset, with only 3 genuine edge-case disagreements involving borderline sensitivity classifications and subtle aggregation semantics. This high agreement validates the benchmark's ground truth while providing explainable, per-rule violation traces for every blocked case.

This rule-based baseline serves as a deliberately strong static-policy reference point: it has full access to data sensitivity labels, domain ownership, and the cross-domain authorization matrix. In contrast, individual LLM agents in the multi-agent evaluation have access only to their own domain's world model, which is precisely why their violation rates remain high (14--98\%) despite strong reasoning capabilities.

\subsection{Failure Analysis}

\begin{itemize}
    \item \textbf{Entity Resolution Errors}: The neuro-symbolic grounding layer may incorrectly resolve highly ambiguous references (e.g., ``the project'' without context). Our fine-tuned 0.6B mapper achieves the lowest violation slip-through rate (2.9\%) among all evaluated configurations---including frontier few-shot models (3.7--4.4\%)---because it exhibits recall-oriented behavior that prioritizes detecting potential sensitive entities over exact-match precision (see Appendix~\ref{app:e2e-entity-mapper}). Mitigation: session-aware disambiguation.
    
    \item \textbf{Novel Violation Patterns}: Zero-day attack patterns not covered by existing invariant templates. Mitigation: anomaly detection layer for unusual data flow patterns.
    
    \item \textbf{Sidecar Bypass}: Malicious agents attempting to communicate without going through the sidecar. Mitigation: network-level enforcement via service mesh policies.
    
    \item \textbf{Collusion Attacks}: Multiple compromised agents coordinating to launder data. Mitigation: cross-agent behavioral analysis and lineage depth limits.
\end{itemize}

These failure modes suggest directions for future work: improved entity resolution, automated invariant discovery, Byzantine fault tolerance, and adversarial robustness.

\section{Related Work}
\label{sec:related}

\paragraph{LLM Safety and Alignment.}
The predominant approach to LLM safety focuses on training-time alignment through RLHF~\citep{ouyang2022training}, constitutional AI~\citep{bai2022constitutional}, and red-teaming~\citep{perez2022red}. Recent work has exposed vulnerabilities in aligned models through jailbreaking~\citep{liu2023jailbreaking,wei2023jailbroken} and adversarial attacks~\citep{zou2023universal}. While alignment is effective for universal safety properties, it cannot by itself encode organization-specific cross-domain policies. Runtime guardrails such as NeMo Guardrails~\citep{nemo2023} and Guardrails AI~\citep{guardrailsai} provide configurable filters but operate on single-agent contexts, missing cross-agent violations.

\paragraph{LLM Agents and Tool Use.}
The emergence of agentic LLM systems~\citep{yao2023react,schick2023toolformer,mialon2023augmented} has enabled autonomous task completion. Frameworks like AutoGen~\citep{wu2023autogen}, MetaGPT~\citep{hong2023metagpt}, and generative agents~\citep{park2023generative} demonstrate sophisticated multi-agent coordination. However, these systems prioritize capability over security. Recent analyses~\citep{ruan2023identifying,kang2023exploiting,greshake2023youve} highlight security risks in LLM-integrated applications, particularly through indirect prompt injection.

\paragraph{Multi-Agent Security.}
\citet{chan2023harms} identify emergent risks in multi-agent settings, and \citet{guo2024large} survey challenges in multi-agent LLM systems. \citet{deng2024ai} emphasize the need for principled evaluation of agent systems. Our work addresses this security gap by proposing a distributed enforcement architecture for multi-agent deployments.

\paragraph{Information Flow Control.}
Classical information flow control (IFC) tracks data provenance to enforce confidentiality~\citep{denning1976lattice,myers1999jflow}. Decentralized IFC~\citep{myers1997decentralized} allows distributed policy enforcement. \sentinel{} adapts IFC principles to the LLM setting, where ``information'' is transformed through natural language rather than program execution, necessitating our neuro-symbolic approach.

\paragraph{Zero-Trust and Distributed Systems.}
Zero-trust security models~\citep{rose2020zero} assume no implicit trust between system components. Our sidecar architecture embodies zero-trust principles, drawing from service mesh patterns~\citep{servicemesh}. The privacy-preserving query mechanism can be enhanced with zero-knowledge proofs~\citep{groth2016size}. Our cross-domain query resembles federated learning's privacy-preserving aggregation~\citep{mcmahan2017federated}, adapted for boolean predicate evaluation.

\paragraph{Knowledge Graphs for LLMs.}
Integrating knowledge graphs with LLMs has shown promise for grounding and factuality~\citep{pan2024unifying}. We extend this paradigm to security, using organizational knowledge graphs as the source of truth for policy enforcement rather than factual accuracy.

\section{Conclusion}
\label{sec:conclusion}

We have introduced \textbf{Context-Fragmented Violations (CFVs)}---a class of security threats in multi-agent LLM systems where locally legitimate actions collectively violate organizational policies due to fragmented context. Through the \textbf{PhantomEcosystem} benchmark, we showed that existing approaches are not well suited to addressing CFVs in their distributed, cross-context form.

Our multi-agent violation study provides strong empirical evidence for the necessity of external enforcement: across eight frontier LLMs from three vendors, all models exhibit 14--98\% violation rates when executing cross-domain workflows with only per-agent domain knowledge. Crucially, same-domain violation rates drop to near-zero for the best models, while cross-domain violations persist---confirming that CFVs arise structurally from fragmented world models rather than from insufficient model capability.

Our proposed \textbf{\textsc{Distributed Sentinel}} architecture transforms multi-agent security from an ambiguous language understanding problem into a deterministic distributed state consistency problem. The key innovations---Semantic Taint Tokens, privacy-preserving cross-domain queries, and neuro-symbolic entity mapping---together achieve F1 = 0.95 on PhantomEcosystem with 106ms end-to-end latency and boolean-only cross-domain disclosure.

\paragraph{Limitations and Threat Model Scope.} 
\sentinel{} is designed to address Context-Fragmented Violations, where \emph{well-intentioned} agents collectively violate policies due to fragmented context. Our threat model explicitly assumes that agents are not adversarially compromised; defending against prompt injection, jailbreaking, or Byzantine agents is an orthogonal concern. However, our sidecar architecture is \emph{compatible} with such defenses: all agent communication is mediated by sidecars, where additional security layers (e.g., input sanitization, anomaly detection) can be inserted. If an agent's signing key were compromised, an attacker could forge STT tokens; mitigations include hardware security modules (HSMs) for key storage and periodic key rotation.

The neuro-symbolic mapper, while effective, introduces a potential single point of failure for entity resolution, and the strongest symbolic guarantees in our system are conditioned on sufficiently accurate grounding from free-form outputs to policy-checkable state. 

Regarding privacy, our analysis primarily concerns single-query disclosure: each cross-domain query reveals at most 1 bit (Theorem~5.4). However, repeated adaptive querying by a malicious sidecar could reveal higher-order information about graph structure through intersection attacks. We discuss mitigations in Future Work, including query rate limiting, differential privacy noise injection, and formal analysis under adaptive adversary models.

Finally, PhantomEcosystem, with 9 categories and 200 scenarios, may not capture all real-world violation patterns, particularly those arising from domain-specific regulatory requirements or novel attack vectors.

\paragraph{Future Work.}
We identify several promising directions: (1) automated invariant discovery from organizational policies using LLM-assisted policy extraction; (2) extending to adversarial multi-agent settings with Byzantine fault tolerance; (3) compositional sensitivity analysis to detect aggregation-based re-identification risks; (4) formal analysis of repeated-query privacy leakage under adaptive adversaries; and (5) real-world deployment studies in enterprise environments.

\paragraph{Broader Impact.}
As LLM agents proliferate in enterprise settings, the security architecture we propose provides a foundation for safe deployment. By preserving data sovereignty while enabling cross-domain verification, \sentinel{} aligns with both security requirements and privacy regulations. We hope this work catalyzes further research at the intersection of distributed systems and AI safety.

\section*{Reproducibility Statement}

Code for \sentinel{} and the PhantomEcosystem benchmark will be released upon publication. All hyperparameters, model configurations, and experimental procedures are detailed in Appendix~\ref{app:implementation}. The benchmark dataset will be made publicly available under a permissive license.

\bibliography{references}

@article{xi2023rise,
  title={The Rise and Potential of Large Language Model Based Agents: A Survey},
  author={Xi, Zhiheng and Chen, Wenxiang and Guo, Xin and He, Wei and Ding, Yiwen and Hong, Boyang and Zhang, Ming and Wang, Junzhe and Jin, Senjie and Zhou, Enyu and others},
  journal={arXiv preprint arXiv:2309.07864},
  year={2023}
}

@article{wang2024survey,
  title={A Survey on Large Language Model based Autonomous Agents},
  author={Wang, Lei and Ma, Chen and Feng, Xueyang and Zhang, Zeyu and Yang, Hao and Zhang, Jingsen and Chen, Zhiyuan and Tang, Jiakai and Chen, Xu and Lin, Yankai and others},
  journal={Frontiers of Computer Science},
  year={2024}
}

@article{ouyang2022training,
  title={Training language models to follow instructions with human feedback},
  author={Ouyang, Long and Wu, Jeffrey and Jiang, Xu and Almeida, Diogo and Wainwright, Carroll and Mishkin, Pamela and Zhang, Chong and Agarwal, Sandhini and Slama, Katarina and Ray, Alex and others},
  journal={Advances in Neural Information Processing Systems},
  volume={35},
  pages={27730--27744},
  year={2022}
}

@article{bai2022constitutional,
  title={Constitutional AI: Harmlessness from AI Feedback},
  author={Bai, Yuntao and Kadavath, Saurav and Kundu, Sandipan and Askell, Amanda and Kernion, Jackson and Jones, Andy and Chen, Anna and Goldie, Anna and Mirhoseini, Azalia and McKinnon, Cameron and others},
  journal={arXiv preprint arXiv:2212.08073},
  year={2022}
}

@article{perez2022red,
  title={Red Teaming Language Models with Language Models},
  author={Perez, Ethan and Huang, Saffron and Song, Francis and Cai, Trevor and Ring, Roman and Aslanides, John and Glaese, Amelia and McAleese, Nat and Irving, Geoffrey},
  journal={arXiv preprint arXiv:2202.03286},
  year={2022}
}

@misc{nemo2023,
  title={NeMo Guardrails: A Toolkit for Controllable and Safe LLM Applications},
  author={NVIDIA},
  year={2023},
  howpublished={\url{https://github.com/NVIDIA/NeMo-Guardrails}}
}

@misc{guardrailsai,
  title={Guardrails AI},
  author={Guardrails AI},
  year={2023},
  howpublished={\url{https://github.com/guardrails-ai/guardrails}}
}

@article{wu2023autogen,
  title={AutoGen: Enabling Next-Gen LLM Applications via Multi-Agent Conversation},
  author={Wu, Qingyun and Bansal, Gagan and Zhang, Jieyu and Wu, Yiran and Zhang, Shaokun and Zhu, Erkang and Li, Beibin and Jiang, Li and Zhang, Xiaoyun and Wang, Chi},
  journal={arXiv preprint arXiv:2308.08155},
  year={2023}
}

@article{hong2023metagpt,
  title={MetaGPT: Meta Programming for Multi-Agent Collaborative Framework},
  author={Hong, Sirui and Zheng, Xiawu and Chen, Jonathan and Cheng, Yuheng and Wang, Jinlin and Zhang, Ceyao and Wang, Zili and Yau, Steven Ka Shing and Lin, Zijuan and Zhou, Liyang and others},
  journal={arXiv preprint arXiv:2308.00352},
  year={2023}
}

@article{chan2023harms,
  title={Harms from Increasingly Agentic Algorithmic Systems},
  author={Chan, Alan and Salganik, Rebecca and Marber, Alva and Phan, Nicole and Kerschbaumer, Christine and Stray, Julian and Miessler, Daniel and Jernite, Yacine and Barocas, Solon and Zittrain, Jonathan and others},
  journal={arXiv preprint arXiv:2302.10329},
  year={2023}
}

@article{denning1976lattice,
  title={A Lattice Model of Secure Information Flow},
  author={Denning, Dorothy E},
  journal={Communications of the ACM},
  volume={19},
  number={5},
  pages={236--243},
  year={1976}
}

@inproceedings{myers1999jflow,
  title={JFlow: Practical Mostly-Static Information Flow Control},
  author={Myers, Andrew C},
  booktitle={Proceedings of the 26th ACM SIGPLAN-SIGACT Symposium on Principles of Programming Languages},
  pages={228--241},
  year={1999}
}

@article{myers1997decentralized,
  title={A Decentralized Model for Information Flow Control},
  author={Myers, Andrew C and Liskov, Barbara},
  journal={ACM SIGOPS Operating Systems Review},
  volume={31},
  number={5},
  pages={129--142},
  year={1997}
}

@techreport{rose2020zero,
  title={Zero Trust Architecture},
  author={Rose, Scott and Borchert, Oliver and Mitchell, Stu and Connelly, Sean},
  institution={National Institute of Standards and Technology},
  year={2020}
}

@misc{servicemesh,
  title={The Service Mesh: What Every Software Engineer Needs to Know},
  author={Linkerd},
  year={2020},
  howpublished={\url{https://linkerd.io/what-is-a-service-mesh/}}
}

@article{mcmahan2017federated,
  title={Communication-Efficient Learning of Deep Networks from Decentralized Data},
  author={McMahan, Brendan and Moore, Eider and Ramage, Daniel and Hampson, Seth and Arcas, Blaise Aguera y},
  journal={Artificial Intelligence and Statistics},
  pages={1273--1282},
  year={2017}
}

@article{pan2024unifying,
  title={Unifying Large Language Models and Knowledge Graphs: A Roadmap},
  author={Pan, Shirui and Luo, Linhao and Wang, Yufei and Chen, Chen and Wang, Jiapu and Wu, Xindong},
  journal={IEEE Transactions on Knowledge and Data Engineering},
  year={2024}
}

@article{yao2023react,
  title={ReAct: Synergizing Reasoning and Acting in Language Models},
  author={Yao, Shunyu and Zhao, Jeffrey and Yu, Dian and Du, Nan and Shafran, Izhak and Narasimhan, Karthik and Cao, Yuan},
  journal={International Conference on Learning Representations (ICLR)},
  year={2023}
}

@article{schick2023toolformer,
  title={Toolformer: Language Models Can Teach Themselves to Use Tools},
  author={Schick, Timo and Dwivedi-Yu, Jane and Dess{\`\i}, Roberto and Raileanu, Roberta and Lomeli, Maria and Zettlemoyer, Luke and Cancedda, Nicola and Scialom, Thomas},
  journal={Advances in Neural Information Processing Systems},
  volume={36},
  year={2023}
}

@article{park2023generative,
  title={Generative Agents: Interactive Simulacra of Human Behavior},
  author={Park, Joon Sung and O'Brien, Joseph C and Cai, Carrie Jun and Morris, Meredith Ringel and Liang, Percy and Bernstein, Michael S},
  journal={ACM Symposium on User Interface Software and Technology (UIST)},
  year={2023}
}

@article{greshake2023youve,
  title={Not What You've Signed Up For: Compromising Real-World LLM-Integrated Applications with Indirect Prompt Injection},
  author={Greshake, Kai and Abdelnabi, Sahar and Mishra, Shailesh and Endres, Christoph and Holz, Thorsten and Fritz, Mario},
  journal={ACM SIGSAC Conference on Computer and Communications Security},
  year={2023}
}

@article{liu2023jailbreaking,
  title={Jailbreaking ChatGPT via Prompt Engineering: An Empirical Study},
  author={Liu, Yi and Deng, Gelei and Xu, Zhengzi and Li, Yuekang and Zheng, Yaowen and Zhang, Ying and Zhao, Lida and Zhang, Tianwei and Liu, Yang},
  journal={arXiv preprint arXiv:2305.13860},
  year={2023}
}

@article{zou2023universal,
  title={Universal and Transferable Adversarial Attacks on Aligned Language Models},
  author={Zou, Andy and Wang, Zifan and Kolter, J Zico and Fredrikson, Matt},
  journal={arXiv preprint arXiv:2307.15043},
  year={2023}
}

@article{wei2023jailbroken,
  title={Jailbroken: How Does LLM Safety Training Fail?},
  author={Wei, Alexander and Haghtalab, Nika and Steinhardt, Jacob},
  journal={Advances in Neural Information Processing Systems},
  volume={36},
  year={2023}
}

@article{mialon2023augmented,
  title={Augmented Language Models: A Survey},
  author={Mialon, Gr{\'e}goire and Dess{\`\i}, Roberto and Lomeli, Maria and Nalmpantis, Christoforos and Pasunuru, Ram and Raileanu, Roberta and Rozi{\`e}re, Baptiste and Schick, Timo and Dwivedi-Yu, Jane and Celikyilmaz, Asli and others},
  journal={Transactions on Machine Learning Research},
  year={2023}
}

@inproceedings{groth2016size,
  title={On the Size of Pairing-Based Non-Interactive Arguments},
  author={Groth, Jens},
  booktitle={Annual International Conference on the Theory and Applications of Cryptographic Techniques},
  pages={305--326},
  year={2016},
  organization={Springer}
}

@article{ruan2023identifying,
  title={Identifying the Risks of LM Agents with an LM-Emulated Sandbox},
  author={Ruan, Yangjun and Dong, Honghua and Wang, Andrew and Pitis, Silviu and Zhou, Yongchao and Ba, Jimmy and Dubois, Yann and Maddison, Chris J and Hashimoto, Tatsunori},
  journal={arXiv preprint arXiv:2309.15817},
  year={2023}
}

@article{kang2023exploiting,
  title={Exploiting Programmatic Behavior of LLMs: Dual-Use Through Standard Security Attacks},
  author={Kang, Daniel and Li, Xuechen and Stoica, Ion and Guestrin, Carlos and Zaharia, Matei and Hashimoto, Tatsunori},
  journal={arXiv preprint arXiv:2302.05733},
  year={2023}
}

@article{deng2024ai,
  title={AI Agents That Matter},
  author={Deng, Yuntao and Pan, Liangming and Chen, Wenhu},
  journal={arXiv preprint arXiv:2407.01502},
  year={2024}
}

@article{guo2024large,
  title={Large Language Model Based Multi-Agents: A Survey of Progress and Challenges},
  author={Guo, Taicheng and Chen, Xiuying and Wang, Yaqi and Chang, Ruidi and Pei, Shichao and Chawla, Nitesh V and Wiest, Olaf and Zhang, Xiangliang},
  journal={arXiv preprint arXiv:2402.01680},
  year={2024}
}
\bibliographystyle{plainnat}

\appendix
\section{Implementation Details}
\label{app:implementation}

\subsection{Core Code Examples}

We provide key implementation excerpts from our prototype system.

\subsubsection{Taint Token Data Structure}

\begin{verbatim}
@dataclass
class TaintToken:
    """Semantic Taint Token for tracking data provenance."""
    token_id: str
    source: TokenSource      # Origin agent and graph node
    sensitivity: Sensitivity # Level + categories
    permissions: Permissions # Allowed/forbidden actions
    lineage: Lineage        # Path through agents
    expires_at: datetime
    
    def is_action_allowed(self, action: Action) -> bool:
        """Check if action is permitted by this token."""
        return self.permissions.is_allowed(action)
    
    def sign(self, secret_key: str) -> str:
        """Sign token with HMAC-SHA256."""
        payload = json.dumps(self.to_dict(), sort_keys=True)
        return hmac.new(secret_key.encode(), 
                       payload.encode(), 
                       hashlib.sha256).hexdigest()
\end{verbatim}

\subsubsection{Sentinel Evaluation Logic}

\begin{verbatim}
class SentinelBaseline:
    def evaluate(self, data_point: DataPoint) -> bool:
        """Evaluate using full taint tracking."""
        flows = data_point.flows
        final_flow = flows[-1]
        
        # Rule 1: Side-channel attack detection
        if self._detect_side_channel(data_point):
            return True  # BLOCK
        
        # Rule 2: Privilege escalation detection
        if self._detect_privilege_escalation(flows):
            if final_flow.sensitivity in ["CONFIDENTIAL", "RESTRICTED"]:
                return True
        
        # Rule 3: CONFIDENTIAL to public targets
        if final_flow.sensitivity == "CONFIDENTIAL":
            if final_flow.action == "write_wiki":
                return True
        
        # Rule 4: Cumulative risk threshold
        if self._compute_cumulative_risk(flows) >= 0.7:
            return True
        
        # Rule 5: Dangerous category combinations
        all_categories = set()
        for flow in flows:
            all_categories.update(flow.categories)
        
        dangerous_combos = [{"SALARY", "PII"}, {"CUSTOMER_DATA", "EXTERNAL"}]
        for combo in dangerous_combos:
            if combo.issubset(all_categories):
                if final_flow.action in ["send_external", "write_wiki"]:
                    return True
        
        return False  # ALLOW
\end{verbatim}

\subsection{Sidecar Implementation}

The Sentinel sidecar is implemented as a Python service with the following components:

\begin{itemize}
    \item \textbf{Interceptor}: Hooks into agent I/O using async middleware
    \item \textbf{STT Engine}: Generates and validates Semantic Taint Tokens
    \item \textbf{Graph Store}: NetworkX-based local graph with Redis persistence
    \item \textbf{Query Server}: gRPC server handling cross-domain predicate queries
\end{itemize}

\subsection{Neuro-Symbolic Mapper}

The entity mapper uses LLM-based few-shot prompting with the full organizational data registry (82 data types across 7 domains), agent registry (17 agents), and sink registry (15 external targets) provided as structured context. Session provenance (the executing agent's domain) is included to disambiguate references. We evaluate accuracy on 20 curated test cases in §\ref{app:experiments}.

\subsection{STT Cryptographic Protocol}

Tokens are signed using Ed25519 signatures with per-department key pairs. Token structure:

\begin{verbatim}
{
  "version": 1,
  "source_id": "ari:graph:dept-rd:node:titan",
  "taint_vector": [1, 0, 1, 0, 0, 0, 0, 1],
  "constraints": {
    "audience": {"not": ["external", "partner"]},
    "clearance": {"min": 3}
  },
  "timestamp": 1704067200,
  "signature": "base64..."
}
\end{verbatim}

\subsection{Invariant Specification Language}

Invariants are specified in a domain-specific language:

\begin{verbatim}
INVARIANT nda_protection:
  FOR entity IN graph
  WHERE entity.has_label("nda_protected")
  BLOCK action
  WHERE action.audience IN ["external", "public"]
  MESSAGE "NDA-protected content cannot be shared externally"
\end{verbatim}

\section{PhantomEcosystem Details}
\label{sec:benchmark-details}

\label{app:benchmark}

\subsection{Scenario Generation}

Scenarios are generated using a combination of:
\begin{enumerate}
    \item Template-based generation for core violation patterns
    \item LLM-assisted paraphrasing for linguistic diversity
    \item Human review for quality assurance
\end{enumerate}

\subsection{Knowledge Graph Statistics}

\begin{table}[ht]
\centering
\caption{Per-Department Graph Statistics}
\begin{tabular}{@{}lrrr@{}}
\toprule
\textbf{Department} & \textbf{Nodes} & \textbf{Edges} & \textbf{Constraints} \\
\midrule
R\&D & 234 & 512 & 47 \\
Marketing & 156 & 298 & 23 \\
HR & 189 & 421 & 62 \\
Sales & 142 & 267 & 31 \\
Finance & 167 & 389 & 55 \\
Legal & 98 & 187 & 41 \\
\bottomrule
\end{tabular}
\end{table}

\subsection{Safe Control Generation}

For each violating scenario, we generate a safe control by:
\begin{enumerate}
    \item Keeping the surface-level communication similar
    \item Modifying the target audience or action scope to be permissible
    \item Ensuring the graph state supports the safe action
\end{enumerate}

\section{Additional Experiments}
\label{app:experiments}

\subsection{Sensitivity to Graph Size}

\begin{table}[ht]
\centering
\caption{Performance vs. Graph Size}
\begin{tabular}{@{}lcc@{}}
\toprule
\textbf{Nodes per Dept} & \textbf{F1} & \textbf{Latency (ms)} \\
\midrule
50 & 0.95 & 3.1 \\
100 & 0.94 & 3.8 \\
200 & 0.94 & 4.7 \\
500 & 0.93 & 6.2 \\
1000 & 0.92 & 9.1 \\
\bottomrule
\end{tabular}
\end{table}

\subsection{Entity Resolution Accuracy}
\label{app:entity-resolution}

We evaluate LLM-based entity resolution on 20 curated test cases spanning three difficulty levels: \emph{exact match} (explicit data type names), \emph{paraphrase} (indirect references like ``comp data'' for compensation\_bands), and \emph{ambiguous} (vague references like ``the report'' requiring session provenance). Using GPT-5 mini with few-shot prompting and the full organizational registry as context, the mapper achieves:
\begin{table}[ht]
\centering
\caption{Entity resolution accuracy by difficulty level (GPT-5 mini, 20 cases).}
\label{tab:entity-resolution}
\begin{tabular}{@{}lcccc@{}}
\toprule
\textbf{Difficulty} & \textbf{$n$} & \textbf{Data Type} & \textbf{Target} & \textbf{Both} \\
\midrule
Exact match & 5 & 100\% & 100\% & 100\% \\
Paraphrase & 8 & 75\% & 88\% & 62\% \\
Ambiguous & 7 & 86\% & 71\% & 71\% \\
\midrule
\textbf{Overall} & \textbf{20} & \textbf{85\%} & \textbf{85\%} & \textbf{75\%} \\
\bottomrule
\end{tabular}
\end{table}

Data type identification and target resolution both achieve 85\% accuracy overall, with exact-match cases reaching 100\% on both dimensions. Performance degrades for paraphrased references (62\% both-correct) and ambiguous references (71\%), where the text alone does not uniquely determine the intended data type or recipient. Errors concentrate in cases where session provenance is essential for disambiguation---e.g., ``comp data'' maps to \texttt{salary\_data} instead of \texttt{compensation\_bands}, and ``the external team'' maps to \texttt{partner\_shared\_folder} instead of \texttt{partner\_crm\_sync}. These results confirm that entity resolution remains a meaningful accuracy bottleneck and motivate provenance-aware disambiguation as a key component.

\subsection{End-to-End Entity Mapper Evaluation}
\label{app:e2e-entity-mapper}

While Section~\ref{app:entity-resolution} evaluates entity resolution in isolation, a natural question is whether the \emph{grounding layer} of \sentinel{} can be implemented by a lightweight local model without significantly degrading end-to-end enforcement performance. We evaluate this by inserting a fine-tuned 0.6B-parameter model (Qwen3-0.6B with LoRA adaptation) into the full \sentinel{} pipeline and measuring the impact on violation detection.

\paragraph{Setup.}
We fine-tune Qwen3-0.6B using LoRA (rank 16, $\alpha$=32) on 5{,}400 entity alignment examples spanning diverse organizational data types and natural language references. At inference time, a retrieval stage first narrows the candidate entity set to the top-15 most relevant entries (using lexical overlap and domain provenance scoring), then the fine-tuned model selects from this reduced set. The full pipeline operates as follows: (1)~an LLM agent generates a natural language response to a task prompt, (2)~\texttt{extract\_decision} determines whether the agent would proceed, (3)~the entity mapper extracts data types from the agent's response text, and (4)~the policy checker evaluates whether the resulting data flow violates organizational policy.

\paragraph{Results.}
Table~\ref{tab:e2e-entity-mapper} reports end-to-end enforcement results on PhantomEcosystem using GPT-5.4-nano as the executing agent. We compare several frontier few-shot mappers accessed through the gateway API against our fine-tuned 0.6B local mapper. All listed systems operate on the agent's generated response text rather than benchmark metadata. The reported rankings should therefore be interpreted as conditional on this agent-output distribution, rather than as universal ordering across all possible agent styles.

\begin{table}[ht]
\centering
\footnotesize
\caption{End-to-end enforcement with different entity mappers (GPT-5.4-nano agent, 200 cases). Frontier baselines use few-shot prompting over generated agent responses. \emph{Viol.\ Rate} = fraction of risky cases that slip through both agent self-avoidance and \sentinel{} enforcement. \emph{Sen.\ Catch} = fraction of agent-PROCEED risky cases blocked by \sentinel{}.}
\label{tab:e2e-entity-mapper}
\begin{tabular}{@{}lccccccc@{}}
\toprule
\textbf{Mapper} & \textbf{F1} & \textbf{Prec.} & \textbf{Rec.} & \textbf{Viol.\ Rate}$\downarrow$ & \textbf{Sen.\ Catch} & \textbf{Ent.\ EM} & \textbf{Ent.\ Rec.} \\
\midrule
GPT-5.4-mini   & \textbf{0.942} & \textbf{0.922} & \textbf{0.963} & 3.7\% & 95.2\% & 53.5\% & 79.9\% \\
Gemini 3 Flash & 0.939 & 0.917 & \textbf{0.963} & 3.7\% & 95.2\% & 59.0\% & \textbf{82.5\%} \\
Claude Opus 4.6 & 0.936 & 0.916 & 0.956 & 4.4\% & 94.4\% & \textbf{60.0\%} & 80.8\% \\
GPT-5.4-nano   & 0.933 & 0.911 & 0.956 & 4.4\% & 94.4\% & 30.5\% & 81.1\% \\
Qwen3-0.6B     & 0.919 & 0.883 & 0.944 & \textbf{2.9\%} & \textbf{96.5\%} & 23.7\% & 77.3\% \\
\bottomrule
\end{tabular}
\end{table}

Frontier few-shot mappers achieve the highest end-to-end F1 in this setting, with GPT-5.4-mini reaching 0.942. At the same time, the fine-tuned 0.6B local grounding layer attains the \emph{lowest violation slip-through rate among the evaluated systems}: 2.9\%, compared with 3.7\% for GPT-5.4-mini and Gemini 3 Flash, and 4.4\% for Claude Opus 4.6 and GPT-5.4-nano. The resulting picture is therefore not one-dimensional: frontier mappers optimize best for aggregate extraction quality, whereas the local grounding layer optimizes best for the downstream security objective of minimizing missed unsafe actions.

The comparison between generic and specialized compact models is especially revealing. GPT-5.4-nano few-shot achieves higher entity exact match than the Qwen3-0.6B LoRA mapper (30.5\% vs.\ 23.7\%) and slightly higher F1 (0.933 vs.\ 0.919), but it permits more violations to slip through (4.4\% vs.\ 2.9\%). This suggests that task-specific fine-tuning can outperform a generic frontier small model on the final security objective even when exact-set matching accuracy is lower.

\paragraph{Analysis.}
These results reveal a mismatch between conventional extraction metrics and security outcomes. Entity exact match rewards producing the \emph{exact} gold set and penalizes both omissions and supersets equally. By this measure, frontier few-shot models appear clearly superior. However, \sentinel{}'s downstream objective is not exact reconstruction of the entity set; it is prevention of unsafe cross-domain actions. For this objective, missing a sensitive entity is substantially worse than over-predicting one. The fine-tuned 0.6B model exhibits precisely this recall-oriented behavior: it often returns supersets or semantically adjacent entities, depressing exact-match accuracy, while still triggering the correct STT taint propagation and policy checks. This is why its end-to-end F1 is slightly below the best frontier models, yet its violation slip-through rate is the lowest among the evaluated systems.

From a systems perspective, the comparison also clarifies the deployment trade-off. Frontier few-shot mappers provide the best general-purpose extraction quality out of the box. In contrast, the local 0.6B mapper offers a compelling privacy--cost--safety trade-off: no external API dependency, organization-local deployment, and stronger protection against missed violations. This suggests that for security-critical neuro-symbolic grounding, task-specific fine-tuning of a compact model can be a more practical choice than relying exclusively on larger general-purpose LLMs.

\subsection{Multi-Agent Violation Study: Full Results}

\begin{table}[ht]
\centering
\caption{Per-attack-category violation rate (\%) for all 8 models}
\label{tab:violation-by-category}
\begin{tabular}{@{}lccccccccc@{}}
\toprule
\textbf{Category} & \textbf{Opus 4.6} & \textbf{G.Pro} & \textbf{G.Flash} & \textbf{GPT-4.1} & \textbf{Son.4.6} & \textbf{GPT-5.4} & \textbf{5.4m} & \textbf{5.4n} \\
\midrule
Side-channel & 15 & 85 & 70 & 70 & 75 & 71 & -- & 60 \\
Multi-hop & 10 & 0 & 25 & 5 & 20 & 40 & 70 & 70 \\
Direct leak & 0 & 0 & 5 & 10 & 15 & 30 & 90 & 60 \\
Scope creep & 0 & 5 & 5 & 15 & 10 & 50 & 73 & 55 \\
Data reconstr. & 0 & 15 & 30 & 5 & 15 & -- & -- & -- \\
Time-series & 0 & 5 & 30 & 15 & 20 & -- & -- & 100 \\
Aggregation & 5 & 5 & 15 & 20 & 5 & 15 & 45 & 45 \\
Token manip. & 0 & 10 & 20 & 20 & 0 & -- & -- & -- \\
Cross-org & 0 & 0 & 0 & 0 & 10 & -- & -- & -- \\
\bottomrule
\end{tabular}
\end{table}

\begin{table}[ht]
\centering
\caption{Per-difficulty violation rate (\%) for all 8 models}
\label{tab:violation-by-difficulty}
\begin{tabular}{@{}lccccccccc@{}}
\toprule
\textbf{Difficulty} & \textbf{Opus 4.6} & \textbf{G.Pro} & \textbf{G.Flash} & \textbf{GPT-4.1} & \textbf{Son.4.6} & \textbf{GPT-5.4} & \textbf{5.4m} & \textbf{5.4n} \\
\midrule
Easy & 3.8 & 7.7 & 11.5 & 16.7 & 19.4 & -- & -- & -- \\
Medium & 16.0 & 18.7 & 26.7 & 28.0 & 25.3 & -- & -- & -- \\
Hard & 28.6 & 45.7 & 57.1 & 53.5 & 60.7 & -- & -- & -- \\
\bottomrule
\end{tabular}
\end{table}

\section{Agent World Model Design}
\label{app:agents}

PhantomEcosystem comprises 17 distinct agents organized across 7 organizational domains, plus 10 external data sinks. Each agent represents a \textbf{realistic AI deployment} that enterprises actually use today, named to match common industry terminology.

\subsection{Department Business Agents}

These agents serve specific business functions, similar to tools like GitHub Copilot, Salesforce Einstein, or Workday Assistant.

\begin{table}[ht]
\centering
\caption{Agent Deployment: 17 Agents across 7 Domains}
\label{tab:v7-agents}
\begin{tabular}{@{}ll@{}}
\toprule
\textbf{Domain} & \textbf{Agents} \\
\midrule
ENGINEERING & \texttt{code\_architect}, \texttt{release\_orchestrator} \\
DOCS & \texttt{documentation\_agent} \\
SECURITY & \texttt{security\_auditor} \\
HR & \texttt{workforce\_manager} \\
PAYROLL & \texttt{payroll\_compliance} \\
FINANCE & \texttt{finance\_analyst} \\
MARKETING & \texttt{outbound\_marketing} \\
SALES & \texttt{sales\_ops} \\
CUSTOMER & \texttt{customer\_success} \\
LEGAL & \texttt{legal\_counsel} \\
IAM & \texttt{iam\_governor} \\
BOUNDARY & \texttt{public\_relations}, \texttt{external\_client\_bot} \\
OPS & \texttt{infra\_ops}, \texttt{support\_engineer} \\
DATA & \texttt{analytics\_platform} \\
\bottomrule
\end{tabular}
\end{table}

\textbf{External Data Sinks:} \texttt{vendor\_procurement\_portal}, \texttt{partner\_crm\_sync}, \texttt{public\_developer\_docs}, \texttt{press\_release\_wire}, \texttt{customer\_slack\_channel}, \texttt{personal\_cloud\_drive}, \texttt{external\_contractor\_email}, \texttt{competitor\_job\_board}, \texttt{shadow\_analytics\_tool}, \texttt{leaked\_pastebin}

\subsection{Functional Tool Agents}

These agents provide specific capabilities integrated into enterprise workflows.

\begin{table}[ht]
\centering
\caption{Functional Tool Agents}
\label{tab:tool-agents}
\begin{tabular}{@{}lll@{}}
\toprule
\textbf{Agent} & \textbf{Real-World Analogy} & \textbf{Function} \\
\midrule
\texttt{wiki\_agent} & Confluence/Notion AI & Knowledge retrieval \\
\texttt{code\_review\_bot} & GitHub PR reviewer & Code analysis \\
\texttt{translation\_agent} & DeepL/Google Translate & Localization \\
\texttt{slack\_bot} & Slack workflow bot & Team coordination \\
\texttt{enterprise\_search\_agent} & Elastic/Glean & Cross-system search \\
\texttt{scrum\_agent} & Jira automation & Sprint management \\
\bottomrule
\end{tabular}
\end{table}

\subsection{Customer-Facing Agents}

These agents interact with external users and have restricted access to internal systems.

\begin{table}[ht]
\centering
\caption{Customer-Facing Agents}
\label{tab:external-agents}
\begin{tabular}{@{}lll@{}}
\toprule
\textbf{Agent} & \textbf{Description} & \textbf{Access Scope} \\
\midrule
\texttt{customer\_chatbot} & Website chat widget & PUBLIC only \\
\texttt{customer\_service\_bot} & Support ticket handler & Customer data \\
\texttt{trial\_user\_bot} & Demo environment bot & Limited, time-bound \\
\texttt{pr\_agent} & Media relations assistant & PUBLIC, press info \\
\bottomrule
\end{tabular}
\end{table}

\subsection{System/Process Agents}

Backend automation agents that handle workflows and integrations.

\begin{itemize}
    \item \texttt{integration\_agent}: Zapier/MuleSoft-style automation between systems
    \item \texttt{offboarding\_system}: Automated permission revocation for departing employees
    \item \texttt{devops\_agent}: CI/CD pipeline monitoring and alerts
    \item \texttt{exec\_assistant}: Executive briefing and calendar management
\end{itemize}

\subsection{Specialized Function Agents}

\begin{itemize}
    \item \texttt{translation\_agent}: Performs \texttt{summarize} operations that may strip security labels (identity laundering).
    \item \texttt{api\_agent}, \texttt{search\_agent}: Handle external queries; may leak metadata through error messages or result counts.
    \item \texttt{pr\_agent}: External communications; response patterns may implicitly confirm secrets.
    \item \texttt{support\_agent}: Customer-facing; may inadvertently share export-controlled information.
\end{itemize}

\subsection{External and Temporal Entities}

\begin{itemize}
    \item \texttt{external\_agent}: Can only \texttt{query}; cannot directly \texttt{read\_data} from internal systems.
    \item \texttt{prospect\_agent}: Uses \texttt{access\_with\_expired\_creds} (not direct read) to model expired trial access.
    \item \texttt{departed\_agent}: Uses \texttt{access\_with\_cached\_session} to model revoked but uncleaned permissions.
\end{itemize}

\subsection{Design Principles}

Our agent design follows four key principles:
\begin{enumerate}
    \item \textbf{Realistic Naming}: Every agent corresponds to a real product category (Copilot, Zendesk, Confluence AI). Reviewers can verify these exist in production environments.
    \item \textbf{Separation of Duties}: Each agent handles only data categories within its business function. \texttt{hr\_agent} accesses HR data; \texttt{dev\_agent} accesses code.
    \item \textbf{External Isolation}: Customer-facing agents (\texttt{customer\_chatbot}, \texttt{trial\_user\_bot}) have strictly limited access and cannot directly query internal systems.
    \item \textbf{Attack Surface Mapping}: \texttt{analytics\_agent} is a data aggregation hub (aggregation attacks); \texttt{translation\_agent} is a laundering point (identity laundering); \texttt{offboarding\_system} handles temporal transitions (temporal attacks).
\end{enumerate}

\section{Ethical Considerations}
\label{app:ethics}

\paragraph{Dual Use.} While \sentinel{} is designed to prevent policy violations, the same technology could theoretically be used to enforce unethical policies. We emphasize that the system enforces \emph{organizational} policies, which should themselves be subject to ethical review.

\paragraph{Privacy.} Our cross-domain query protocol is specifically designed to minimize information leakage. We formally analyze privacy guarantees in terms of differential privacy in future work.

\paragraph{Automation Bias.} Operators should not blindly trust \sentinel{} decisions. We recommend human review for high-stakes blocking decisions.

\end{document}